\documentclass[11pt]{article}
\usepackage{fullpage,amsmath,amsthm,color}
\usepackage{palatino}
\usepackage{graphicx}
\usepackage{colortbl}
\usepackage{enumerate}
\usepackage{algorithm}
\newtheorem{theorem}{Theorem}[section]
\newtheorem{lemma}[theorem]{Lemma}
\newcommand{\colora}{}
\newcommand{\colorb}{}
\newcommand{\reviewa}[1]{}
\newcommand{\reviewb}[1]{}
\newcommand{\Xomit}[1]{}

\title{Replacement Paths via Row Minima of Concise Matrices\thanks{Accepted to {\em SIAM Journal on Discrete Mathematics}.}}

\author{Cheng-Wei Lee\thanks{Email:~{\tt{r99922035@ntu.edu.tw}}.
Department of Computer Science and Information Engineering, National
Taiwan University.}
\and 
Hsueh-I Lu\thanks{Email:~{Corresponding author. \tt{hil@csie.ntu.edu.tw}}.
  Web:~{\tt{www.csie.ntu.edu.tw/\~{}hil}}. Department of Computer
  Science and Information Engineering, National Taiwan University.
  This author also holds joint appointments in the Graduate Institute
  of Networking and Multimedia and the Graduate Institute of
  Biomedical Electronics and Bioinformatics, National Taiwan
  University. Address: 1 Roosevelt Road, Section 4, Taipei 106,
  Taiwan, ROC.  Research supported in part by NSC grant
  101--2221--E--002--062--MY3.}
}

\date{\today}

\begin{document}
\maketitle

\begin{abstract}
Matrix $M$ is {\em $k$-concise} if the finite entries of each column
of $M$ consist of $k$ or less intervals of identical numbers. We give
an $O(n+m)$-time algorithm to compute the row minima of any
$O(1)$-concise $n\times m$ matrix.  Our algorithm yields the first
$O(n+m)$-time reductions from the replacement-paths problem on an
$n$-node $m$-edge undirected graph (respectively, directed acyclic
graph) to the single-source shortest-paths problem on an $O(n)$-node
$O(m)$-edge undirected graph (respectively, directed acyclic graph).
That is, we prove that the replacement-paths problem is no harder than
the single-source shortest-paths problem on undirected graphs and
directed acyclic graphs.  Moreover, our linear-time reductions lead to
the first $O(n+m)$-time algorithms for the replacement-paths problem
on the following classes of $n$-node $m$-edge graphs (1) undirected
graphs in the word-RAM model of computation, (2) undirected planar
graphs, (3) undirected minor-closed graphs, and (4) directed acyclic
graphs.
\end{abstract}

\section{Introduction}
Computing a shortest path between two nodes in a graph is one of the
most fundamental algorithmic problems in computer science. The variant
of the shortest-path problem which asks for a shortest path between
two nodes that avoids a failed node or edge has also been extensively
studied in the last few decades.  
Let $G$ be a graph. 
{\colorb\reviewb{3}\label{b3}For any node $v$ of $G$, let $G-v$
  denote the graph obtained from $G$ by deleting $v$ and its incident
  edges.  For any edge $e$ of $G$, let~$G-e$ denote the graph obtained
  from $G$ by deleting $e$.}
For any subgraph $G'$ of $G$, let $w(G')$ be the sum of edge weights
of $G'$.  An {\em $rs$-path} is a path from node $r$ to node~$s$.  The
{\em distance}~$d_G(r,s)$ from $r$ to $s$ in $G$ is the minimum of
$w(P)$ over all $rs$-paths $P$ of $G$.  A {\em shortest} $rs$-path $P$
of $G$ satisfies $w(P)=d_G(r,s)$.
We study the following two versions of the {\em replacement-paths
  problem} on $G$ with respect to a given shortest $rs$-path $P$ of
$G$:
\begin{itemize}
\item The {\em edge-avoiding version} computes $d_{G-e}(r,s)$ for all
  edges~$e$ of $P$.
\item The {\em node-avoiding version} computes $d_{G-v}(r,s)$ for all
  nodes~$v$ of $P$ other than $r$ and $s$.
\end{itemize}
The edge-avoiding version can be reduced in linear time to the
node-avoiding version: Let $G'$ be the graph obtained from $G$ by
subdividing each edge $xy$ of $P$ into two edges $xv$ and $vy$ with
$w(xv)=w(vy)=w(xy)/2$.  We have $d_{G-xy}(r,s)=d_{G'-v}(r,s)$.  No
linear-time reduction for the other direction is known.
See,~e.g.,~\cite{HershbergerS01,ByersW84,NisanR99} for applications of
the problem.  Extensive surveys for the long history of algorithms and
applications of this problem can be found
in~\cite{EmekPR10,RodittyZ05}.
We show that the replacement-paths problem on an $n$-node $m$-edge
undirected graph can be reduced in $O(n+m)$ time to the single-source
shortest-paths problem on an $O(n)$-node $O(m)$-edge undirected graph.
\begin{theorem}
\label{theorem:theorem1}
Let $G$ be an $n$-node $m$-edge undirected graph.  Let~$P$ be a given
shortest $rs$-path of $G$, where $r$ and $s$ are two distinct nodes of
$G$. Given distances $d_G(r,v)$ and $d_G(v,s)$ for all nodes $v$ of
$G$, we have the following statements.
\begin{enumerate}
\item 
It takes $O(n+m)$ time to solve the edge-avoiding replacement-paths
problem on $G$ with respect to $P$.
\item 
The node-avoiding replacement-paths problem on $G$ with respect to $P$
can be reduced in $O(n+m)$ time to the problem of computing distances
$d_{G_0}(r_0,v)$ for some node $r_0$ and all nodes $v$ of an
$O(n)$-node $O(m)$-edge undirected graph $G_0$.
\end{enumerate}
\end{theorem}

\begin{table}
\caption{Previous work and our results on the replacement-paths problem.}\label{table:table1}
\begin{center}
{\small\begin{tabular}{|c||c|c||c|}
\hline
&edge-avoiding version&node-avoiding version&ours\\
\hline
\hline
 directed graph&$O(mn+n^2\log\log n)$~\cite{GotthilfL09}&$O(mn+n^2\log n)$~\cite{CormenLRS09}&\\
\hline
directed acyclic graph&$O(m+n\cdot\alpha(m,n))$~\cite{Bhosle05}&&$O(m+n)$\\
\hline
directed acyclic graph (RAM)&
$O(m+n\cdot\alpha(2n,n))$~\cite{Bhosle05}&&$O(m+n)$\\
\hline
 undirected graph&$O(m+n\log n)$~\cite{MalikMG89}&$O(m+n\log n)$~\cite{NardelliPW03}&$O(m+n\log n)$\\
\hline
undirected graph (RAM) 
&$O(m\cdot\alpha(m,n))$~\cite{NardelliPW01} &&$O(m+n)$\\
\hline
undirected planar graph&
$O(n)$~\cite{Bhosle05}&&$O(n)$\\\hline
undirected minor-closed graph&
&&$O(n)$\\
\hline
\end{tabular}
}
\end{center}
\end{table}

Combining with Dijkstra's single-source shortest-paths algorithm~(see,
e.g.,~\cite{CormenLRS09}), Theorem~\ref{theorem:theorem1} solves the
replacement-paths problem in $O(m+n\log n)$ time, matching the best
known result for the edge-avoiding version of Malik, Mittal, and
Gupta~\cite{MalikMG89} and that for the node-avoiding version of
Nardelli, Proietti, and Widmayer~\cite{NardelliPW03}.
Combining with the algorithm of Henzinger, Klein, Rao, and
Subramanian~\cite{HenzingerKRS97}, Theorem~\ref{theorem:theorem1}
yields an $O(n+m)$-time algorithm for both versions of the problem on
planar graphs, while $O(n+m)$-time algorithms on planar graphs were
only known for the edge-avoiding version (see Bhosle~\cite{Bhosle05}).
Combining with the algorithm of Tazari and
M{\"u}ller-Hannemann~\cite{TazariM09}, Theorem~\ref{theorem:theorem1}
leads to the first $O(n+m)$-time algorithm on minor-closed graphs.
Combining with the algorithms of Thorup~\cite{Thorup00,Thorup99},
Theorem~\ref{theorem:theorem1} solves both versions of the problem in
$O(n+m)$ time in the word-RAM model of computation, improving upon the
$O(m\cdot\alpha(m,n))$-time transmuter-based algorithm of Nardelli,
Proietti, and Widmayer~\cite{NardelliPW01}, which works only for the
edge-avoiding version.  See~\cite{Pettie08} for more results of the
single-source shortest-paths problem that can be combined with our
reductions to yield efficient algorithms for the replacement-paths
problem.

Our proof of Theorem~\ref{theorem:theorem1} also holds for directed
acyclic graphs.  Since the single-source shortest-paths problem can be
solved in linear time on directed acyclic graphs (see,
e.g.,~\cite{CormenLRS09}), we solve both versions of the
replacement-paths problem on any $n$-node $m$-edge directed acyclic
graph in $O(n+m)$ time, improving upon the algorithm of
Bhosle~\cite{Bhosle05} for the edge-avoiding version, which runs in
$O(m+n\cdot\alpha(2n,n))$ time in the word-RAM model of computation
and runs in $O(m\cdot\alpha(m,n))$ time in general.
\begin{theorem}
\label{theorem:theorem2}
For any two nodes $r$ and $s$ of an $n$-node $m$-edge directed acyclic
graph $G$, it takes $O(n+m)$ time to solve the replacement-paths
problem on $G$ with respect to any given shortest $rs$-path of $G$.
\end{theorem}

Table~\ref{table:table1} compares our results with previous work.

\subsection{Technical overview}

\begin{figure}[t]
\begin{center}
\scalebox{0.83}{
\begin{minipage}[b]{8.5cm}
\begin{center}
\begin{tabular}{|c||c|c|c|c|c|c|c|c|c|c|}
\hline
$M$&$v_0v_6$&$v_0v_8$&$v_6v_7$&$v_6v_5$&$v_7v_4$&$v_9v_5$\\
\hline
\hline
1&13&15&  &  &  &\\
\hline
2&  &15&18&12&  &\\
\hline
3&  &15&  &12&16&\\
\hline
4&  &  &  &12&16&9\\
\hline
5&  &  &  &12&  &9\\
\hline
\end{tabular}\\[10pt]
(a)
\end{center}
\end{minipage}
\begin{minipage}[b]{7.5cm}
\begin{center}
\begin{tabular}{|c||c|c|c|c|c|c|c|c|c|c|}
\hline
$N$&$v_0v_8$&$v_6v_7$&$v_6v_5$&$v_7v_4$&$v_9v_5$\\
\hline
\hline
1&15&19&13&  &\\
\hline
2&15&  &12&20&\\
\hline
3&  &  &12&16&19\\
\hline
4&  &  &12&  &9\\
\hline
\end{tabular}\\[24pt]
(b)
\end{center}
\end{minipage}
}
\end{center}
\caption{(a) A concise $5\times 6$ matrix $M$. (b) A $2$-concise
  $4\times 5$ matrix $N$.  The $\infty$-entries in $M$ and $N$ are
  left out.}
\label{figure:figure1}
\end{figure}

A matrix $M$ is {\em $k$-concise} if the finite entries of each column
of $M$ consist of $k$ or less intervals of identical numbers.  A
$1$-concise matrix is {\em concise}.  Figure~\ref{figure:figure1}(a)
shows a concise matrix.  Figure~\ref{figure:figure1}(b) shows a
$2$-concise matrix.  A $k$-concise matrix may not be sparse, but each
column of a $k$-concise matrix can be concisely represented by $O(k)$
numbers, i.e., three numbers for each of the $k$ or less intervals of
identical finite numbers: (a) the starting row index, (b) the ending
row index, and (c) the identical number of the interval.  For
instance, the columns with indices $v_6v_5$, $v_7v_4$, and $v_9v_5$ of
the $2$-concise matrix in Figure~\ref{figure:figure1}(b) can be
represented by $\langle 1,1,13;2,4,12\rangle$, $\langle 2,2,20;
3,3,16\rangle$, and $\langle 3,3,19;4,4,9\rangle$, respectively.
Throughout the paper, all matrices are in this {\em concise
  representation}.  The {\em row-minima problem} on a matrix $M$ is to
compute the minimum of each row of $M$.  We show that the
replacement-paths problem on an $n$-node $m$-edge undirected
(respectively, directed acyclic) graph can be reduced in $O(n+m)$ time
to the row-minima problem on a $2$-concise $n\times m$ matrix {\colorb\reviewb{5}\label{b5}obtainable} from the solution to the single-source
shortest-paths problem on an $O(n)$-node $O(m)$-edge undirected
(respectively, directed acyclic) graph (see Lemma~\ref{lemma:lemma1}
in~\S\ref{subsection:reduction-edge} for the edge-avoiding version and
Lemma~\ref{lemma:lemma2} in~\S\ref{subsection:reduction-node} for the
node-avoiding version).  
Our reductions {\colorb\reviewb{6}\label{b6}exploit} the structure
properties of replacement paths studied by, e.g., Malik et
al.~\cite{MalikMG89}, Nardelli et
al.~\cite{NardelliPW03,NardelliPW01}, and Bhosle~\cite{Bhosle05}.
\begin{figure}[t]
\begin{center}
\scalebox{0.83}{
\begin{minipage}{7.5cm}
\begin{center}
\begin{tabular}{|c||c|c|c|c|c|c|c|c|c|c|}
\hline
$N_1$&$v_0v_8$&$v_6v_7$&$v_6v_5$&$v_7v_4$&$v_9v_5$\\
\hline
\hline
1&15&19&13&  &  \\
\hline
2&15&  &  &20&  \\
\hline
3&  &  &  &  &19  \\
\hline
4&  &  &  &  & \\
\hline
\end{tabular}\\[24pt]
(a)
\end{center}
\end{minipage}
\begin{minipage}{7.5cm}
\begin{center}
\begin{tabular}{|c||c|c|c|c|c|c|c|c|c|c|}
\hline
$N_2$&$v_0v_8$&$v_6v_7$&$v_6v_5$&$v_7v_4$&$v_9v_5$\\
\hline
\hline
1&  &  &  &  &\\
\hline
2&  &  &12&  &\\
\hline
3&  &  &12&16&\\
\hline
4&  &  &12&  &9\\
\hline
\end{tabular}\\[24pt]
(b)
\end{center}
\end{minipage}
}
\end{center}
\caption{Two concise $4\times 5$ matrices $N_1$ and $N_2$ whose
  entry-wise minimum is the $2$-concise $4\times 5$ matrix $N$ of
  Figure~\ref{figure:figure1}(b). The $\infty$-entries of $N_1$ and
  $N_2$ are left out.}
\label{figure:figure2}
\end{figure}
To show that the replacement-paths problem is no harder than the
single-source shortest-paths problem, we give the first $O(n+m)$-time
algorithm for the row-minima problem on any $O(1)$-concise $n\times m$
matrix (see Lemma~\ref{lemma:lemma3} in~\S\ref{section:row-minima}).
As illustrated by Figure~\ref{figure:figure2}, for any $k$-concise
$n\times m$ matrix $N$ with $k=O(1)$, it takes $O(m)$ time to derive
concise $n\times m$ matrices $N_1,N_2,\ldots,N_k$ whose entry-wise
minimum is $N$.  Thus, the main technical challenge lies in computing
the row minima of an $n\times m$ concise matrix $M$ in $O(n+m)$ time.
The rest of the overview elaborates on our $O(n+m)$-time algorithm for
the row-minima problem on any concisely represented $n\times m$
concise matrix $M$.

The {\em thickness} $\theta$ of $M$ is the length of a longest
interval of identical finite entries over all columns of $M$.  For
instance, the thickness of the matrix in
Figure~\ref{figure:figure1}(a) (respectively,
Figures~\ref{figure:figure2}(a) and~\ref{figure:figure2}(b)) is $4$
(respectively, $2$ and $3$).
The {\em broadness} $\beta$ of $M$ is the minimum of (i) the number of
distinct starting row indices for the intervals of finite entries over
all columns of $M$, and (ii) the number of distinct ending row indices
for the intervals of finite entries over all columns of $M$. For
instance, the broadness of the matrix in
Figure~\ref{figure:figure1}(a) (respectively,
Figures~\ref{figure:figure2}(a) and~\ref{figure:figure2}(b)) is $4$
(respectively, $3$ and $2$).
The row minima of $M$ can be computed in $O(n+m+\theta\cdot\beta)$
time by Lemma~\ref{lemma:lemma6} in~\S\ref{subsection:step1}.  The
thickness and broadness of $M$ can both be as large as $n$, so
applying Lemma~\ref{lemma:lemma6} on $M$ may require $\Omega(n^2)$
time.  Our $O(n+m)$-time algorithm is based upon the technique of
deriving matrices with smaller thickness or broadness whose row minima
yield the row minima of $M$.  (Details are in the proofs of
Lemma~\ref{lemma:lemma3} in~\S\ref{subsection:step3} and
Lemma~\ref{lemma:lemma7} in~\S\ref{subsection:step1}.)  Specifically,
we derive four $n$-row matrices $M_0,M_1,M_2,M_3$ from $M$ according
to some positive integral {\em brush factor} $h$ such that the row
minima of $M$ is the entry-wise minima of the row minima of the four
matrices.
A column of $M$ is {\em $h$-brushed} if it contains at least one
finite entry in rows $h,2h,\ldots,\lfloor\frac{n}{h}\rfloor\cdot h$.
{\colorb\reviewb{8}\label{b8}For instance, all columns of the matrix in
  Figure~\ref{figure:figure3}(a) are $3$-brushed.}
Matrix $M_0$ is the submatrix of $M$ induced by the non-$h$-brushed
columns.
{\colorb See Figure~\ref{figure:figure4}(a) for a matrix $M_0$
  that has no $3$-brushed columns.}
{\colorb\reviewb{9}\label{b9} Matrices $M_1$, $M_2$, and $M_3$ represent the
  $h$-brushed columns of $M$:} Matrix $M_1$ takes over the first $h$
or less finite entries of each $h$-brushed column of $M$ up to the
first {\colora row with a finite entry}\reviewa{3}\label{a3} whose index is an integral
multiple of $h$; matrix $M_3$ takes over the last $h-1$ or less finite
entries of each $h$-brushed column of $M$ starting from the
{\colora row with a finite entry} that immediately succeeds the
last row whose index is an integral multiple of $h$; and matrix $M_2$
takes over the finite entries of each $h$-brushed column in between.
The entry-wise minimum of matrices $M_1$, $M_2$, and $M_3$ is the
submatrix of $M$ induced by the $h$-brushed columns.
\begin{figure}[t]
\begin{center}
\scalebox{0.83}{
\begin{minipage}{6.5cm}
\begin{center}
\begin{tabular}{|c||c|c|c|c|c||c|c|c|}
\hline
$M$&1&2&3&4&5&$\mu$&$q$&$z$\\
\hline
\hline
1 & 9&  &  &  &&   9&  9&  1\\
\hline
2 & 9& 7&  &  &&   7&  7&  1\\
\hline
\rowcolor[gray]{0.9}3& 9& 7& 5&  &&\cellcolor[gray]{1}5&\cellcolor[gray]{1}5&\cellcolor[gray]{1}1\\
\hline
4 &  & 7& 5&  &&   5&   &  0\\
\hline
5 &  & 7& 5& 6&&   5&   &  0\\
\hline
\rowcolor[gray]{0.9}6&  & 7& 5& 6&&\cellcolor[gray]{1}5&\cellcolor[gray]{1}&
\cellcolor[gray]{1}0\\
\hline
7 &  & 7& 5& 6& 8& 5&   &  0\\
\hline
8 &  & 7& 5& 6& 8& 5&   &  0\\
\hline
\rowcolor[gray]{0.9}9&  & 7&  & 6& 8&\cellcolor[gray]{1}6&\cellcolor[gray]{1}6&\cellcolor[gray]{1}1\\
\hline
10&  &  &  & 6& 8&  6&  &  0\\
\hline
11&  &  &  &  & 8&  8& 8&  1\\
\hline
\end{tabular}\\[10pt]
(a)
\end{center}
\end{minipage}
\begin{minipage}{5cm}
\begin{center}
\begin{tabular}{|c||c|c|c|c|c|c|c|}
\hline
$M_1$&1&2&3&4&5\\
\hline
\hline
1 & 9&  &  &  &  \\
\hline
2 & 9& 7&  &  &  \\
\hline
3 & 9& 7& 5&  &  \\
\hline
4 &  &  &  &  &  \\
\hline
5 &  &  &  & 6&  \\
\hline
6 &  &  &  & 6&  \\
\hline
7 &  &  &  &  & 8\\
\hline
8 &  &  &  &  & 8\\
\hline
9 &  &  &  &  & 8\\
\hline
10&  &  &  &  &  \\
\hline
11&  &  &  &  &  \\
\hline
\end{tabular}\\[10pt]
(b)
\end{center}
\end{minipage}
\begin{minipage}{3.4cm}
\begin{center}
\begin{tabular}{|c||c|c|c|c|c|c|}
\hline
$M_2$&2&3&4\\
\hline
\hline
1 &  &  &  \\
\hline
2 &  &  &  \\
\hline
3 &  &  &  \\
\hline
4 & 7& 5&  \\
\hline
5 & 7& 5&  \\
\hline
6 & 7& 5&  \\
\hline
7 & 7&  & 6\\
\hline
8 & 7&  & 6\\
\hline
9 & 7&  & 6\\
\hline
10&  &  &  \\
\hline
11&  &  &  \\
\hline
\end{tabular}\\[10pt]
(c)
\end{center}
\end{minipage}
\begin{minipage}{3.6cm}
\begin{center}
\begin{tabular}{|c||c|c|c|c|c|c|}
\hline
$M_3$&3&4&5\\
\hline
\hline
1 &  &  &  \\
\hline
2 &  &  &  \\
\hline
3 &  &  &  \\
\hline
4 &  &  &  \\
\hline
5 &  &  &  \\
\hline
6 &  &  &  \\
\hline
7 & 5&  &  \\
\hline
8 & 5&  &  \\
\hline
9 &  &  &  \\
\hline
10&  & 6& 8\\
\hline
11&  &  & 8\\
\hline
\end{tabular}\\[10pt]
(d)
\end{center}
\end{minipage}
}
\end{center}
\caption{Each column of matrix $M$ is $3$-brushed.  The minima array
  $\mu$ of $M$ and its corresponding query array~$q$ and auxiliary
  binary string $z$ are displayed to the right of $M$. The entries of
  $q$ that do not matter are left out.  Matrix $M_1$ has thickness $3$
  and broadness $3$.  Every three consecutive rows of $M_2$ are
  identical.  Matrix $M_3$ has thickness $2$ and broadness $2$.  The
  $\infty$-entries in these four matrices are left out.  Matrix $M_1$
  has no all-$\infty$ columns.  The all-$\infty$ columns of $M_2$ and
  $M_3$ are omitted.}
\label{figure:figure3}
\end{figure}
See Figures~\ref{figure:figure3}(b)--\ref{figure:figure3}(d) for the
$M_1$, $M_2$, and $M_3$ obtained from the $M$ in
Figure~\ref{figure:figure3}(a) with brush factor $h=3$.  Matrices
$M_1$ and $M_3$ have thickness $O(h)$ and broadness $O(\frac{n}{h})$,
so the row minima of $M_1$ and $M_3$ can be computed in $O(n+m)$ time
by Lemma~\ref{lemma:lemma6} for any choice of $h$.  In order to
compute the row minima of $M_0$ and $M_2$ in $O(n+m)$ time, we let
$h=\Theta(\log\log n)$ and resort to two intermediate algorithms for
the row-minima problem. As to be explained in the next two paragraphs,
we (1) apply the first intermediate algorithm on an
$O(\frac{n}{h})$-row $O(m)$-column matrix obtained from $M_2$ by
condensing its identical rows and (2) apply the second intermediate
algorithm on $O(h)$-row matrices derived from $M_0$ whose overall
number of rows (respectively, columns) is $O(n)$ (respectively,
$O(m)$).

{\colorb\reviewb{7}\label{b7-1} The
 broadness of the matrix $M_2$ obtained in the previous
  paragraph is $O(\frac{n}{h})$.}
Although the thickness of $M_2$ could be $\Omega(n)$, every $h$
consecutive rows of $M_2$ are identical. See
Figure~\ref{figure:figure3}(c) for an example of $M_2$ with $h=3$.  We
condense matrix $M_2$ into an $O(\frac{n}{h})$-row $O(m)$-column
matrix $M^*$.  By $h=\Theta(\log\log n)$, it takes $O(n+m)$ time to
compute the row minima of $M_2$ by applying our first intermediate
algorithm (see Lemma~\ref{lemma:lemma4} in~\S\ref{subsection:step1})
on the condensed matrix $M^*$.  For the rest of the paragraph, let $M$
(with slight abuse of notation) be the input $n\times m$ matrix of
this $O(m+n\log\log n)$-time intermediate algorithm, which is based
upon the above technique of reducing thickness and broadness in a more
complicated manner.  We first partition $M$ into submatrices
$M_1,M_2,\ldots,M_\ell$ with $\ell=O(\log\log n)$ in $O(m+n\log\log
n)$ time.  Specifically, let $h_0,h_1,\ldots,h_\ell$ be a decreasing
sequence of positive integers such that $h_0\geq n$, $h_1<n$,
$h_\ell=1$, and $h_{k-1}=\Theta(h_k^2)$ holds for each
$k=1,2,\ldots,\ell$.  Let $M_k$ be the submatrix of $M$ induced by the
$h_k$-brushed columns that are not $h_{k-1}$-brushed, implying that
$M_k$ has thickness $O(h_{k-1})=O(h_k^2)$. For each $n\times m_k$
matrix $N=M_k$ with $1\leq k\leq\ell$, we derive three $n\times m_k$
matrices $N_1$, $N_2$, and $N_3$ {\colorb\reviewb{10}\label{b10}with brush factor
  $h=h_k$} (again, as in the proof of Lemma~\ref{lemma:lemma7}
in~\S\ref{subsection:step1} and as illustrated by
Figure~\ref{figure:figure3}).  Both $N_1$ and $N_3$ have thickness
$O(h_k)$ and broadness $O(\frac{n}{h_k})$.  {\colorb\reviewb{11}\label{b11}Since every
  $h_k$ consecutive rows of $N_2$ are identical and $N_2$ are not
  $h_{k-1}$-brushed,} we condense $N_2$ into an $O(\frac{n}{h_k})$-row
$m_k$-column matrix $N^*$ with thickness $O(h_k)$ and broadness
$O(\frac{n}{h_k})$.  The row minima of $N_1$, $N^*$, and $N_3$ can be
computed in $O(n+m_k)$ time by Lemma~\ref{lemma:lemma6}.
The row minima of $M_k=N$ can be obtained from those of $N_1$, $N^*$,
and $N_3$ in $O(n)$ time.
Taking entry-wise minima on the row minima of $M_1,M_2,\ldots,M_\ell$,
we have the row minima of $M$ in time $\sum_{1\leq
  k\leq\ell}O(m_k+n)=O(m+n\log\log n)$.

\begin{figure}[t]
\begin{center}
\scalebox{0.83}{
\begin{minipage}{5cm}
\begin{center}
\begin{tabular}{|c||c|c|c|c|c|c|c|}
\hline
$M_0$&1&2&3&4&5\\
\hline
\hline
1 & 9&  &  &  &  \\
\hline
2 & 9&  &  &  &  \\
\hline
\rowcolor[gray]{0.9}3&  &  &  &  &  \\
\hline
4 &  & 7& 5&  &  \\
\hline
5 &  &  & 5&  &  \\
\hline
\rowcolor[gray]{0.9}6&  &  &  &  &  \\
\hline
7 &  &  &  & 6&  \\
\hline
8 &  &  &  &  &  \\
\hline
\rowcolor[gray]{0.9}9&  &  &  &  &  \\
\hline
10&  &  &  &  & 8\\
\hline
11&  &  &  &  & 8\\
\hline
\end{tabular}\\[10pt]
(a)
\end{center}
\end{minipage}
\begin{minipage}{2.5cm}
\begin{center}
\begin{tabular}{|c||c|c|c|c|c|c|c|}
\hline
$M_1$&1 \\
\hline
\hline
1 & 9\\
\hline
2 & 9\\
\hline
\end{tabular}\\[4.8cm]
(b)
\end{center}
\end{minipage}
\begin{minipage}{3cm}
\medskip
\bigskip
\bigskip
\bigskip
\begin{center}
\begin{tabular}{|c||c|c|c|c|c|c|}
\hline
$M_2$&2&3\\
\hline
\hline
4 & 7& 5\\
\hline
5 &  & 5\\
\hline
\end{tabular}\\[3.3cm]
(c)
\end{center}
\end{minipage}
\begin{minipage}{2.5cm}
\medskip
\medskip
\bigskip
\bigskip
\bigskip
\bigskip
\bigskip
\bigskip
\begin{center}
\begin{tabular}{|c||c|c|c|c|c|c|}
\hline
$M_3$&4\\
\hline
\hline
7 & 6\\
\hline
8 &  \\
\hline
\end{tabular}\\[1.8cm]
(d)
\end{center}
\end{minipage}
\begin{minipage}{2.5cm}
\medskip
\medskip
\medskip
\bigskip
\bigskip
\bigskip
\bigskip
\bigskip
\bigskip
\bigskip
\bigskip
\bigskip
\begin{center}
\begin{tabular}{|c||c|c|c|c|c|c|}
\hline
$M_4$&5\\
\hline
\hline
10& 8\\
\hline
11& 8\\
\hline
\end{tabular}\\[10pt]
(e)
\end{center}
\end{minipage}
}
\end{center}
\caption{$M_0$ has no $3$-brushed columns.  The row minima of $M_0$
  can be obtained from combining the row minima of $M_1$, $M_2$,
  $M_3$, and $M_4$. The $\infty$-entries in the matrices are left
  out.}
\label{figure:figure4}
\end{figure}

{\colorb\reviewb{7}\label{b7-2}\reviewb{12}\label{b12}The
thickness of the matrix $M_0$ obtained in the
  paragraph preceding the previous paragraph is $O(h)$.}  
Since $M_0$ has no $h$-brushed columns, one can partition the finite
entries of $M_0$ into $O(h)$-row matrices $M_1,M_2,\ldots,M_\ell$ with
$\ell=O(\frac{n}{h})$ whose overall number of columns is $O(m)$.  See
Figure~\ref{figure:figure4} for an illustration.  (Details are in the
proof of Lemma~\ref{lemma:lemma3} in~\S\ref{subsection:step3}).
{\colorb\reviewb{1}\label{b1} Recursively applying the procedure
  described in the previous two paragraphs on $M_1,\ldots,M_\ell$
  would only lead to an $O((m+n)\log^* n)$-time algorithm.  Instead,}
by
$h=\Theta(\log\log n)$, the row minima of each $O(h)$-row $m_k$-column
matrix $M_k$ can be computed in $O(m_k+\log\log n)$ time by our second
intermediate algorithm (i.e., Algorithm~\ref{algorithm:algorithm1} in
the proof of Lemma~\ref{lemma:lemma8}
in~\S\ref{subsection:step2}). Putting together the row minima of
$M_1,M_2,\ldots,M_\ell$, we solve the row-minima problem on $M_0$ in
time $\sum_{1\leq k\leq \ell}O(m_k+\log\log n)=O(m+n)$.  This
$O(m_k+\log\log n)$-time intermediate algorithm for the row-minima
problem on any $O(\log\log n)$-row $m_k$-column matrix $M_k$ proceeds
iteratively with the help of two data structures.  For each
$j=1,2,\ldots,m_k$, at the end of the $j$-th iteration, the first data
structure keeps the minimum of the first $j$ columns of each row in a
concise manner such that the minima of consecutive rows can be
efficiently updated. Specifically, let $\mu(i)$ be the minimum of the
first $j$ entries of row $i$.  An array $q$ and a binary string $z$
satisfying $q(\textit{pred}(z,i))=\mu(i)$ for all row indices $i$ are
used to represent array $\mu$, where $\textit{pred}(z,i)$ denotes the
largest index $i_1$ with $i_1\leq i$ and $z(i_1)=1$.  The value of
$\mu(i)$ can be obtained from $q(\textit{pred}(z,i))$.  Updating
$\mu(i)$ for all indices $i$ with $\textit{pred}(z,i)=i_1$ to a
smaller value can be done by decreasing $q(i_1)$.  See
Figure~\ref{figure:figure3}(a) for an example of $\mu$, $q$, and $z$
with $j=5$.  If the index $\textit{pred}(z,i)$ for each $i$ were
$O(1)$-time computable and the value of $z(i)$ for each $i$ were
$O(1)$-time updatable, then our Algorithm~\ref{algorithm:algorithm1}
in~\S\ref{subsection:step2} would have been an $O(n+m)$-time algorithm
for the row-minima problem on any $n\times m$ matrix.
However, it is impossible in general to come up with a
polynomial-sized dynamic data structure for binary string $z$ that
supports both $O(1)$-time update on $z(i)$ and $O(1)$-time query
$\textit{pred}(z,i)$~\cite{BeameF02}.
Fortunately, the binary string $z$ needed to represent the minima
array $\mu$ of the $O(h)$-row matrix $M_k$ has only $O(h)=O(\log\log
n)$ bits. Thus, one can pre-compute all possible updates and queries
on $z$ in $o(n)$ time and organize all the pre-computed information in
an $o(n)$-space table capable of supporting each query and update on
$z$ in $O(1)$ time. With the help of this second data structure, our
second intermediate algorithm computes the row minima of each $M_k$
with $1\leq k\leq\ell$ in $O(m_k+\log\log n)$ time.

\subsection{Related work}

On directed graphs with nonnegative weights, Gotthilf and
Lewenstein~\cite{GotthilfL09} gave the best known algorithm, running
in $O(mn+n^2\log\log n)$ time, for the edge-avoiding version of the
replacement-paths problem.  The $O(mn+n^2\log n)$-time algorithm of
running Dijkstra's $O(m+n\log n)$-time algorithm for $O(n)$ times
remains the best known algorithm for the node-avoiding version.
Bernstein~\cite{Bernstein10} gave an algorithm to output
$(1+\epsilon)$-approximate solutions for both versions of
the problem for any positive parameter $\epsilon$.
Hershberger, Suri, and Bhosle~\cite{HershbergerSB07} showed a lower
bound $\Omega(m\sqrt{n})$ on the time complexity of the problem in the
path-comparison model of Karger, Koller, and
Phillips~\cite{KargerKP93}.
The randomized algorithm of Roditty and Zwick~\cite{RodittyZ05} on
unweighted directed graphs runs in $\tilde{O}(m\sqrt{n})$ time. On
directed graphs with integral weights in $\{-W,\ldots,W\}$, Weimann
and Yuster~\cite{WeimannY10,WeimannY13} gave an
$\tilde{O}(Wn^{\omega}+W^{2/3}n^{1+2\omega/3})$-time randomized
algorithm for both versions of the problem, where $\omega$ is the
infimum of all numbers such that multiplying two $n\times n$ matrices
takes $\tilde{O}(n^{\omega})$ time.  The running time was improved to
$\tilde{O}(Wn^{\omega})$ by Vassilevska~Williams~\cite{Williams11},
who~\cite{Williams12} recently reduced the long-standing upper bound
on $\omega$ of Coppersmith and Winograd~\cite{CoppersmithW90} from
$\omega<2.376$ to $\omega<2.3727$.
Recently, 
Grandoni and Vassilevska~Williams~\cite{GrandoniV12}
addressed the single-source version of the problem.
On directed planar graphs with nonnegative weights, the algorithm of
Wulff-Nilsen~\cite{Wulff-Nilsen10} runs in $O(n\log n)$ time,
improving on the $O(n\log^3 n)$-time algorithm of Emek, Peleg, and
Roditty~\cite{EmekPR10} and the $O(n\log^2 n)$-time algorithm of
Klein, Mozes, and Weimann~\cite{KleinMW10}.  Erickson and
Nayyeri~\cite{EricksonN11} extended Wulff-Nilsen's result on
bounded-genus graphs.

Bernstein and Karger~\cite{Bernstein09} addressed the all-pairs
replacement-paths problem by giving an $\tilde{O}(n^2)$-space
$\tilde{O}(mn)$-time data structure capable of answering
$d_{G-v}(r,s)$ for any nodes $r$, $s$, and $v$ of directed graph $G$
in $O(1)$ time. Baswana, Lath, and Mehta~\cite{BaswanaLM12} studied
the single-source and all-pairs replacement-paths problems on directed
planar graphs.  Malik et al.~\cite{MalikMG89} studied replacement
paths that avoid multiple failed edges.
Duan and Pettie~\cite{DuanP09} studied replacement paths that avoid
two failed nodes or edges.  Weimann et al.~\cite{WeimannY13} studied
replacement paths that avoid multiple failed nodes and edges.
Chechik, Langberg, Peleg, and Roditty~\cite{ChechikLPR12} studied near
optimal replacement paths that avoid multiple failed edges.

For the closely related problem of finding $k$ shortest $rs$-paths for
any given nodes $r$ and $s$ of directed graph $G$ with nonnegative
edge weights, Eppstein~\cite{Eppstein98} gave an $O(m+n\log n+k)$-time
algorithm, which may output non-simple paths.  If the output paths are
required to be simple, the best currently known algorithm, also due to
Gotthilf et al.~\cite{GotthilfL09}, uses replacement paths.
Specifically, Roditty and Zwick~\cite{RodittyZ05} showed that the
problem can be reduced to $O(k)$ computations of the second shortest
simple $rs$-path. Therefore, the replacement-paths algorithm of
Gotthilf et al.~yields an $O(kmn + kn^2\log\log n)$-time algorithm
for the problem of finding $k$ shortest simple paths.
See~\cite{Roditty07, Bernstein10, HershbergerMS07} for more results on
this related problem.
See~\cite{KleinMW10,AtallahK92,
  AggarwalKMSW87,NakanoO98,MozesW10,BradfordR96,RamanV94,KaplanMNS12}
for results involving the row-minima problem on matrices with special
structures.

\subsection{Road map}

The rest of the paper is organized as follows.
Section~\ref{section:prelim} gives the preliminaries, including our
$O(n+m)$-time reductions for both versions of the replacement-paths
problem on an $n$-node $m$-edge undirected graph to (1) the
row-minima problem on $O(1)$-concise $n\times m$ matrices and (2)
the single-source shortest-paths problem on $O(n)$-node
$O(m)$-edge undirected graphs.  Both reductions also
work for directed acyclic graphs.  Section~\ref{section:row-minima}
gives our $O(n+m)$-time algorithm for the row-minima problem on any
$O(1)$-concise $n\times m$ matrix and proves
Theorems~\ref{theorem:theorem1} and~\ref{theorem:theorem2}.
Section~\ref{section:conclude} concludes the paper.

\section{Preliminaries}
\label{section:prelim}

\begin{figure}[t]
\centerline{\scalebox{0.83}{\small\input{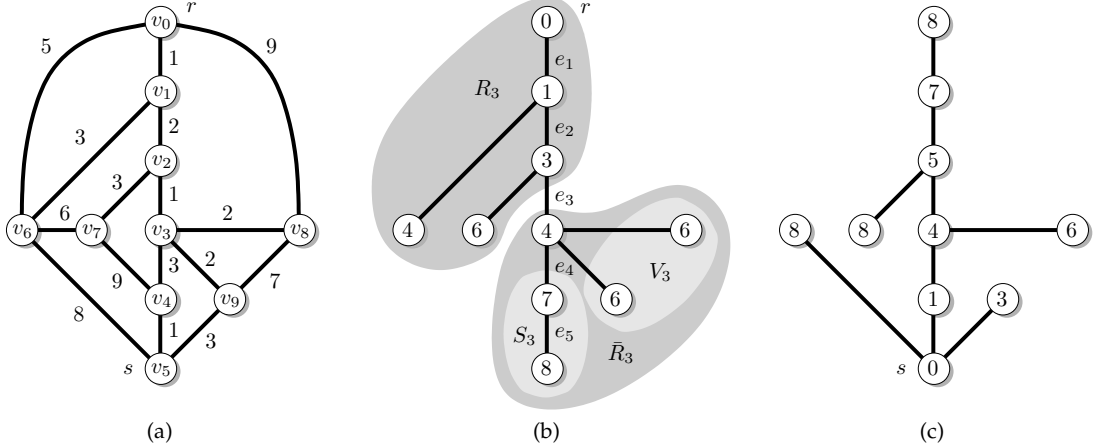}}}
\caption{(a) Graph $G$ in which $(v_0,v_1,\ldots,v_5)$ is a shortest
  $rs$-path $P$.  (b) A shortest-paths tree $T$ of $G$ rooted at
  $r$, in which $P$ consists of edges $e_1,e_2,\ldots,e_5$. The number in each node is its distance from $r$ in $G$.  (c) A
  shortest-paths tree $T'$ of $G$ rooted at $s$. The number in each
  node is its distance to $s$ in $G$.}
\label{figure:figure5}
\end{figure}

Let $|S|$ denote the cardinality of set $S$.  A row (respectively,
column) of a matrix is {\em dummy} if all of its entries are $\infty$.
Given distances $d_G(r,v)$ for all nodes $v$ of an $n$-node $m$-edge
graph $G$, a shortest-paths tree $T$ in $G$ rooted at $r$ that
contains the given shortest $rs$-path $P$ can be obtained in $O(m+n)$
time.  Let $p$ be the number of edges in $P$.  Let $v_0,
v_1,\ldots,v_p$ be the nodes of $P$ from $r=v_0$ to $s=v_p$.  For each
$i=1,2,\ldots,p$, let $e_i$ be edge $v_{i-1}v_i$.  See
Figures~\ref{figure:figure5}(a) and~\ref{figure:figure5}(b) for an
example of $G$, $T$, and $P$.

Subsection~\ref{subsection:reduction-edge} gives our reduction for the
edge-avoiding version.  Subsection~\ref{subsection:reduction-node}
gives our reduction for the node-avoiding version.  
Our reductions are presented in a way that also works for directed
acyclic graphs.
The reductions for directed acyclic graphs hold even with the
existence of negative-weighted edges, while the reductions for
undirected graphs assume nonnegative edge weights.  We comment on
handling negative weights for undirected graphs
in~\S\ref{section:conclude}.

\subsection{A reduction for the edge-avoiding version}
\label{subsection:reduction-edge}

For each node $v$ of $G$, let {\em level} $\lambda(v)$ of $v$ in $T$
be the largest index $i$ such that $v_i$ is on the path of $T$ from
$r$ to $v$.  Levels $\lambda(v)$ for all nodes $v$ of $G$ can be
computed from $T$ in $O(n)$ time.  For each $i=1,2,\ldots,p$,
\begin{itemize}
\item let $R_i$ consist of the nodes $x$
with $\lambda(x)\leq i-1$ and
\item 
let $\bar{R}_i$ consist of the nodes $y$
with $\lambda(y)\geq i$.  
\end{itemize}
That is, $R_i$ (respectively, $\bar{R}_i$) consists of the nodes $v$
that are reachable (respectively, unreachable) from $r$ in $T-e_i$.
See Figure~\ref{figure:figure5}(b) for an illustration of $R_i$ and
$\bar{R}_i$.  For any edge $xy$ of $G$ with $\lambda(x)<\lambda(y)$,
define
\begin{displaymath}
\textit{replacement-cost}_1(x,y)=d_G(r,x)+w(xy)+d_G(y,s).
\end{displaymath}
Since $R_i$ and $\bar{R}_i$ define a cut between nodes $r$ and $s$,
any $rs$-path of $G$ contains some edge $xy$ with $x\in R_i$ and $y\in
\bar{R}_i$.  
We have
\begin{equation}
d_{G-e_i}(r,s)=\min\{\textit{replacement-cost}_1(x,y) \mid \mbox{$x\in
  R_i$, $y\in \bar{R}_i$, and $xy\in G-e_i$}\}
\label{eq:eq1}
\end{equation}
for each $i=1,2,\ldots,p$ (see also,
e.g.,~\cite{NardelliPW01,MalikMG89}).  The {\em edge-replacement
  matrix} of $G$ with respect to $T$ and $P$ is the $p\times m$ matrix
$M$ defined by
\begin{displaymath}
M(i,xy)=
\left\{
\begin{array}{ll}
\textit{replacement-cost}_1(x,y)&\mbox{if $\lambda(x) < i \leq \lambda(y)$ and $e_i \ne xy$}\\
\infty&\mbox{otherwise},
\end{array}
\right.
\end{displaymath}
\reviewb{13}\label{b13-1}%
for each $i=1,2,\ldots,p$ and each edge $xy$ of $G$ with
$\lambda(x) {\colorb\ <\ } \lambda(y)$.  For instance, the matrix in
Figure~\ref{figure:figure1}(a) is the edge-replacement matrix of the
graph $G$ in Figure~\ref{figure:figure5}(a) with respect to the tree
$T$ and path $P$ in Figure~\ref{figure:figure5}(b), where the dummy
columns are omitted.  Let $G'$ be the graph obtained from $G$ by
reversing the direction of each edge of $G$. (This statement handles
the case that $G$ is a directed acyclic graph. For the undirected
case, we simply have $G=G'$.)  
{\colorb\reviewb{14}\label{b14}The distances $d_G(v,s)$ for all
nodes $v$ of $G$ and
a shortest-paths tree $T'$ in $G'$ rooted at $s$
can be obtained from each other in $O(m+n)$ time.}  See Figure~\ref{figure:figure5}(c) for
an example of $T'$.

\begin{lemma}
\label{lemma:lemma1}
The edge-replacement matrix $M$ of $G$ with respect to $T$ and $P$ is
a concise matrix whose concise representation can be obtained from
$G$, $P$, $T$, and $T'$ in $O(n+m)$ time.  Moreover, for each
$i=1,2,\ldots,p$, the minimum of the $i$-th row of $M$ equals
$d_{G-e_i}(r,s)$.
\end{lemma}

\begin{proof}
By definition of $M$, if the $xy$-th column of $M$ is not dummy, then
$\lambda(x)\ne \lambda(y)$. Let $x$ and $y$ be the endpoints of such
an edge with $\lambda(x)<\lambda(y)$.  The entries of the $xy$-th
column in rows $\lambda(x)+1,\lambda(x)+2,\ldots, \lambda(y)$ are all
$\textit{replacement-cost}_1(x,y)$. The other entries are all
$\infty$.  Since each column of $M$ consists of at most one interval
of identical finite numbers, $M$ is concise.  Given $G$, $P$, $T$, and
$T'$, values $\textit{replacement-cost}_1(x,y)$ for all edges $xy$ of
$G$ with $\lambda(x)<\lambda(y)$ can be obtained in overall $O(n+m)$
time.  Matrix $M$ can be obtained from $G$, $P$, $T$, and $T'$ in
$O(n+m)$ time.  The minimum of the $i$-th row is the minimum of
$\textit{replacement-cost}_1(x,y)$ over all edges $xy$ of $G$ with
$\lambda(x)<i\leq\lambda(y)$ and $e_i\ne xy$. By definition of $R_i$
and $\bar{R}_i$, edge $xy$ satisfies $\lambda(x)<i\leq\lambda(y)$ if
and only if $x\in R_i$ and $y\in\bar{R}_i$. By
Equation~(\ref{eq:eq1}), the minimum of the $i$-th row of $M$ is
indeed $d_{G-e_i}(r,s)$.  The lemma is proved.
\end{proof}

\subsection{A reduction for the node-avoiding version}
\label{subsection:reduction-node}
Observe that the level $\lambda(v)$ of node $v$ in $T$ is also the
smallest index $i$ such that $v$ is reachable from $r$ in $T-v_{i+1}$.
For each $i=1,\ldots,p-1$, let the nodes of $G-v_i$ be partitioned
into $R_i$, $V_i$, and $S_i$, where
\begin{itemize}
\item 
{\colorb\reviewb{15}\label{b15}$R_i$, as defined in
  \S\ref{subsection:reduction-edge}, consists of the nodes $x$ with
  $\lambda(x)\leq i-1$,}

\item $V_i$ consists of the nodes $x\ne v_i$ with $\lambda(x)=i$, and
\item $S_i$ consists of the nodes $y$ with $\lambda(y)>i$.
\end{itemize}
See Figure~\ref{figure:figure5}(b) for an illustration of $R_i$, $V_i$, and
$S_i$, 
{\colorb where $V_i$ and $S_i$ are depicted\reviewb{16}\label{b16}
 by lighter shaded regions}.
Since $R_i\cup V_i$ and $S_i$ define a cut for nodes $r$ and $s$ in
$G-v_i$, each $rs$-path of $G-v_i$ contains some edge $xy$ with $x\in
R_i\cup V_i$ and $y\in S_i$.  For any node subset $U$ of $G$, let
$G[U]$ denote the subgraph of $G$ induced by $U$.
We have
\begin{equation}
\begin{array}{rcl}
d_{G-v_i}(r,s)
&=&
\min\{d_{G[R_i\cup V_i]}(r,x) + w(xy) + d_{{\colorb G}}(y,s)
\mid \mbox{$x\in R_i \cup V_i$, $y\in S_i$, $xy\in G$}\}\\
&=&\min\{\\
& &\quad\min\{d_G(r,x) + w(xy) + d_G(y,s)
\mid \mbox{$x\in R_i$, $y\in S_i$, $xy\in G$}\},\\
& &\quad\min\{d_{G[R_i\cup V_i]}(r,x) + w(xy) + d_{G}(y,s)
\mid \mbox{$x\in V_i$, $y\in S_i$, $xy\in G$}\}\\
& &\},
\end{array}
\label{eq:eq2}
\end{equation}
where {\colorb\reviewb{2}\label{b2}the first equality is proved by Nardelli et 
al.~\cite[Lemma~3]{NardelliPW03}} and
the second equality follows from the observation that
$d_{G[R_i\cup V_i]}(r,x) = d_G(r,x)$ holds for each node $x\in R_i$.

\begin{figure}[t]
\centerline{\scalebox{0.83}{\small\input{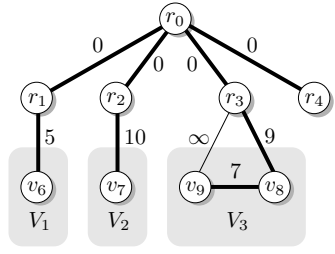}}}
\caption{The graph $G_0$ obtained from the graph $G$ in
  Figure~\ref{figure:figure5}(a) and the tree $T$ and path $P$ in
  Figure~\ref{figure:figure5}(b).  The edges in thick lines form a
  shortest-paths tree $T_0$ of $G_0$ rooted at $r_0$.}
\label{figure:figure6}
\end{figure}

We now define a graph $G_0$ and specify a node $r_0$ of $G_0$ such
that
\begin{equation}
d_{G[R_i\cup V_i]}(r,x)=d_{G_0}(r_0,x)
\label{eq:eq3}
\end{equation}
holds for each $i=1,2,\ldots,p-1$ and each node $x\in V_i$.  For each
$i=1,2,\ldots,p-1$, let $G_i$ be $G[V_i]$ plus one new node $r_i$ and
$|V_i|$ new edges, where for each node $x\in V_i$ the $x$-th new edge
is $r_ix$ with weight
\begin{displaymath}
w(r_ix)=\min\{d_G(r,u)+w(ux)\mid
u\in R_i, ux\in G\}.
\end{displaymath}
Let graph $G_0$ be $G_1\cup G_2\cup\cdots \cup G_{p-1}$ plus a new
node $r_0$ and $p-1$ zero-weighted edges
$r_0r_1,r_0r_2,\ldots,r_0r_{p-1}$.  
{\colora $G_0$ is the disjoint union of $p-1$ induced subgraphs of $G$
  plus a tree with internal nodes $r_0,r_1,\ldots,r_{p-1}$.  
  \reviewa{1}\label{a1}
  For the case that $G$ is a directed acyclic graph, all edges of the
  tree are outgoing toward the disjoint union of the $p-1$ induced
  subgraphs of $G$, which is acyclic. $G_0$ has to be a directed
  acyclic graph.  For the case that $G$ is planar, the disjoint union
  of the $p-1$ induced subgraphs of $G$ is planar.  If edge $r_ix$ for
  some node $x\in V_i$ has finite edge weight, $x$ has at least one
  neighbor of $G$ in $R_i$. Although $G_0$ may not be planar, the
  subgraph of $G_0$ induced by the edges with finite edge weights has
  to be planar.}
Let $T_0$ be a shortest-paths tree
of $G_0$ rooted at $r_0$.  See Figure~\ref{figure:figure6} for an
example.  Observe that $G_0$ is an $O(n)$-node $O(m)$-edge graph,
obtainable in $O(n+m)$ time from $G$ and $T$, such that
Equation~(\ref{eq:eq3}) holds for each $i=1,2,\ldots,p-1$.  For any
edge $xy$ of $G$ with $\lambda(x)<\lambda(y)$, define
\begin{displaymath}
\textit{replacement-cost}_2(x,y)=d_{G_0}(r,x)+w(xy)+d_G(y,s).
\end{displaymath}
The {\em node-replacement matrix} of $G$ with respect to $T$ and $P$
is the $(p-1)\times m$ matrix $N$ defined by
\begin{displaymath}
N(i,xy)\left\{
\begin{array}{ll}
\textit{replacement-cost}_2(x,y)&
\mbox{if $\lambda(x)=i<\lambda(y)$ and $x\ne v_i$}\\
\textit{replacement-cost}_1(x,y)&
\mbox{if $\lambda(x)<i<\lambda(y)$ and $x\ne v_i$}\\ 
\infty&\mbox{otherwise},
\end{array}
\right.
\end{displaymath}
\reviewb{13}\label{b13-2}
for each $i=1,2,\ldots,p-1$ and each edge $xy$ of $G$ with $\lambda(x)
{\colorb\ <\ } \lambda(y)$.  For instance, the matrix in
Figure~\ref{figure:figure1}(b) is the node-replacement matrix of the
graph $G$ in Figure~\ref{figure:figure5}(a) with respect to the tree
$T$ and path $P$ in Figure~\ref{figure:figure5}(b), where the dummy
columns are omitted.

\begin{lemma}
\label{lemma:lemma2}
The node-replacement matrix $N$ of $G$ with respect to $T$ and $P$ is
a $2$-concise matrix whose concise representation can be obtained from
$G$, $P$, $T$, $T'$, and $T_0$ in $O(n+m)$ time. Moreover, for each
$i=1,2,\ldots,p-1$, the minimum of the $i$-th row of $N$ equals
$d_{G-v_i}(r,s)$.
\end{lemma}
\begin{proof}
By definition of $N$, if the $xy$-th column of $N$ with
$\lambda(x)\leq \lambda(y)$ is not dummy, then $\lambda(x)+1\leq
\lambda(y)$.  The entry of the $xy$-th column in row $\lambda(x)$ is
$\textit{replacement-cost}_2(x,y)$. If $\lambda(x)+2\leq\lambda(y)$,
the entries of the $xy$-th column in rows
$\lambda(x)+1,\lambda(x)+2,\ldots,\lambda(y)-1$ are all
$\textit{replacement-cost}_1(x,y)$. The other entries of the $xy$-th
column are all $\infty$.  Since the finite entries of each column of
$N$ consists of at most two intervals of identical numbers, $N$ is
$2$-concise.  Given $G$, $P$, $T$, $T'$, and $T_0$, values
$\textit{replacement-cost}_1(x,y)$ and
$\textit{replacement-cost}_2(x,y)$ for all edges $xy$ of $G$ with
$\lambda(x)<\lambda(y)$ can be obtained in overall $O(n+m)$
time. Matrix $N$ can be obtained from $G$, $P$, $T$, $T'$, and $T_0$
in $O(n+m)$ time.  By Equations~(\ref{eq:eq2}) and~(\ref{eq:eq3}), we
have
\begin{eqnarray*}
d_{G-v_i}(r,s)&=&\min\{\\
&&\quad\min\{\textit{replacement-cost}_1(x,y) \mid \mbox{$x\in R_i$, $y\in S_i$, $xy\in G$}\},\\
&&\quad\min\{\textit{replacement-cost}_2(x,y) \mid \mbox{$x\in V_i$, $y\in S_i$, $xy\in G$}\}\\
&&\}.
\end{eqnarray*}
For each $i=1,\ldots,p-1$, the minimum of the $i$-th row of $N$ is
indeed $d_{G-v_i}(r,s)$.  The lemma is proved.
\end{proof}

\section{\texorpdfstring{The row minima of an $\boldsymbol{O(1)}$-concise matrix in linear time}{The row minima of an O(1)-concise matrix in linear time}}
\label{section:row-minima}

This section proves Lemma~\ref{lemma:lemma3}.
Theorem~\ref{theorem:theorem1} follows immediately from
Lemmas~\ref{lemma:lemma1},~\ref{lemma:lemma2}, and~\ref{lemma:lemma3}.
Theorem~\ref{theorem:theorem2} follows immediately from
Lemma~\ref{lemma:lemma3} and the analogous versions of
Lemmas~\ref{lemma:lemma1} and~\ref{lemma:lemma2} for directed acyclic
graphs.

\begin{lemma}
\label{lemma:lemma3}
It takes $O(n+m)$ time to compute the row minima of a concisely
represented $O(1)$-concise $n\times m$ matrix.
\end{lemma}

As illustrated in Figure~\ref{figure:figure2}, a $k$-concise $n\times
m$ matrix $M$ with $k=O(1)$ can be decomposed in $O(m)$ time into $k$
concise $n\times m$ matrices whose entry-wise minimum is $M$.  To
prove Lemma~\ref{lemma:lemma3}, it suffices to solve the row-minima
problem on any $n\times m$ concise matrix in $O(n+m)$ time.
For the rest of the section, all matrices are concise. Each matrix $M$
is concisely represented by arrays $a_M$, $b_M$, and $c_M$ such that,
for each $i=1,2,\ldots,n$ and each $j=1,2,\ldots,m$, the $(i,j)$-entry
of $M$ can be determined in $O(1)$ time by
\begin{displaymath}
M(i,j)=
\left\{
\begin{array}{ll}
c_M(j)&\mbox{if $a_M(j)\leq i\leq b_M(j)$}\\
\infty&\mbox{otherwise}.
\end{array}
\right.
\end{displaymath}
For instance, if $M$ is the matrix in Figure~\ref{figure:figure3}(a),
then $a_M=(1,2,3,5,7)$, $b_M=(3,9,8,10,11)$, and
$c_M=(9,7,5,6,8)$.  Subscripts $M$ in $a_M$, $b_M$, and $c_M$ can be
omitted, if matrix $M$ is clear from the context.

Subsection~\ref{subsection:step1} proves Lemma~\ref{lemma:lemma4},
which states an $O(m+n\log\log n)$-time algorithm for solving the
row-minima problem on any $n\times m$
matrix. Subsection~\ref{subsection:step2} proves
Lemma~\ref{lemma:lemma8}, which states an $O(m+\log\log n)$-time
algorithm for solving the row-minima problem on any $O(\log\log
n)\times m$ matrix, with the help of an $O(n)$-time pre-computable
$O(n)$-space data structure that supports $O(1)$-time queries and
updates on any $O(\log\log n)$-bit binary string.
Subsection~\ref{subsection:step3} proves Lemma~\ref{lemma:lemma3}
using Lemmas~\ref{lemma:lemma4} and~\ref{lemma:lemma8}.

\subsection{A near-linear-time intermediate algorithm}
\label{subsection:step1}

\begin{lemma}
\label{lemma:lemma4}
It takes $O(m+n\log\log n)$ time to compute the row minima of an
$n\times m$ matrix.
\end{lemma}

This subsection proves Lemma~\ref{lemma:lemma4}, which requires
Lemmas~\ref{lemma:lemma5},~\ref{lemma:lemma6},
and~\ref{lemma:lemma7}. An $n\times m$ matrix $M$ is {\em sorted} if
the following properties hold, where (a) $M_i$ is the submatrix of $M$
induced by the columns whose indices $j$ satisfy $a_M(j)=i$, and (b)
$m_i$ is the number of columns in $M_i$.
\begin{enumerate}[\em Property~S1:]
\item 
\label{property:sorted1}
$a_M(1)\leq a_M(2)\leq\cdots\leq a_M(m)$.
\item 
\label{property:sorted2}
$b_{M_i}(1)\leq b_{M_i}(2)\leq \cdots \leq b_{M_i}(m_i)$ holds for
each $i=1,\ldots,n$.
\end{enumerate}
That is, if $M$ is sorted, then $(a_M(1),b_M(1)),
(a_M(2),b_M(2)),\ldots, (a_M(m),b_M(m))$ are in lexicographically
non-decreasing order.  For instance, the matrices $M$ in
Figures~\ref{figure:figure1}(a) and~\ref{figure:figure3}(a), 
the matrix $M_0$ in Figure~\ref{figure:figure4}(a), 
and the matrix $M_9$ in
Figure~\ref{figure:figure7} are sorted.  The matrix $N_1$ in
Figure~\ref{figure:figure2} is not sorted, since the column with index
$v_0v_8$ is not the third column.

\begin{lemma}
\label{lemma:lemma5}
It takes $O(n+m)$ time to reorder the columns of an $n\times m$ matrix
such that the resulting matrix is sorted.
\end{lemma}

\begin{proof}
Since $a(j)$ and $b(j)$ for all indices $j=1,2,\ldots,m$ are positive
integers in $\{1,2,\ldots,n\}$, the lemma is straightforward by
counting sort~(see, e.g.,~\cite{CormenLRS09}).
\end{proof}

\begin{figure}[t]
\begin{center}
\scalebox{0.83}{
\begin{tabular}{|c||c|c|c|c|c|c|c|c|c|c||c|}
\hline
$M_9$& 1& 2& 3& 4& 5& 6& 7& 8& 9&10&row minimum\\
\hline
\hline
9    &{\cellcolor[gray]{0.9}\it\color{blue} 3}&95&25&66&32&76&51&88&76&81&{\cellcolor[gray]{0.9}
3}\\
\hline
10   &  &{\it\color{blue}95}&25&66&32&76&51&88&76&81&{\cellcolor[gray]{0.9}25}\\
\hline
11   &  &  &{\cellcolor[gray]{0.9}\it\color{blue} 25}&66&32&76&51&88&76&81&{\cellcolor[gray]{0.9}25}\\
\hline
12   &  &  &  &{\it\color{blue}66}&32&76&51&88&76&81&{\cellcolor[gray]{0.9}32}\\
\hline
13   &  &  &  &  &{\cellcolor[gray]{0.9}\it\color{blue} 32}&76&51&88&76&81&{\cellcolor[gray]{0.9}32}\\
\hline
14   &  &  &  &  &{\cellcolor[gray]{0.9}\it\color{blue} 32}&76&51&88&76&81&{\cellcolor[gray]{0.9}32}\\
\hline
15   &  &  &  &  &  &{\it\color{blue} 76}&{\cellcolor[gray]{0.9}\it\color{blue}51}&88&76&81&{\cellcolor[gray]{0.9}51}\\
\hline
16   &  &  &  &  &  &  &  &{\it\color{blue} 88}&76&81&{\cellcolor[gray]{0.9}76}\\
\hline
17   &  &  &  &  &  &  &  &  &{\cellcolor[gray]{0.9}\it\color{blue} 76}&{\it\color{blue}81}&{\cellcolor[gray]{0.9}76}\\
\hline
\end{tabular}
}
\end{center}
\caption{A sorted $n$-row $m_i$-column thickness-$\theta$ matrix $M_i$
  with $n=17$, $i=9$, $m_i=10$, and $\theta=9$. The dummy rows of
  $M_i$ are omitted. The $\infty$-entries are left out. The italic entries form the
  lower-left boundary of the finite entries.}
\label{figure:figure7}
\end{figure}

Define
\begin{eqnarray*}
\textit{thickness}(M)&=&\max\{b_M(j)-a_M(j)+1\mid 1\leq j\leq m\};\\
\textit{broadness}(M)&=&\min\{
|\{a_M(1),a_M(2),\ldots,a_M(m)\}|,
|\{b_M(1),b_M(2),\ldots,b_M(m)\}|
\}.
\end{eqnarray*}
For instance, we have $\textit{thickness}(M)=\textit{broadness}(M)=4$
for the matrix $M$ in Figure~\ref{figure:figure1}(a) and
$\textit{thickness}(M_9)=9$ and $\textit{broadness}(M_9)=1$ for the
matrix $M_9$ in Figure~\ref{figure:figure7}.

\begin{lemma}
\label{lemma:lemma6}
It takes $O(n+m+\textit{thickness}(M)\cdot\textit{broadness}(M))$ time
to compute the row minima of an $n\times m$ matrix $M$.
\end{lemma}

\begin{proof}
Let $\theta=\textit{thickness}(M)$ and $\beta=\textit{broadness}(M)$.
Subscripts $M$ of $a_M$ and $b_M$ in the proof are omitted.  We prove
the lemma for the case with $\beta=|\{a(1),a(2),\ldots,a(m)\}|$.  The
case with $\beta=|\{b(1),b(2),\ldots,b(m)\}|$ can be proved by
reversing the row order of $M$.  We first apply
Lemma~\ref{lemma:lemma5} to have $M$ sorted in $O(n+m)$ time.  For
each $i=1,2,\ldots,n$, let $M_i$ be the submatrix of $M$ induced by
columns whose indices $j$ satisfy $a(j)=i$. Let $m_i$ be the number of
columns in $M_i$.  For each of the $\beta$ indices $i$ with $m_i \geq
1$, the non-dummy rows of submatrix $M_i$ are all in rows
$i,i+1,\ldots,i+\theta-1$.  Since $a(j)=i$ holds for all column
indices $j$ of $M_i$, the sequence of minima of rows
$i,i+1,\ldots,i+\theta-1$ of $M_i$ is non-decreasing.  By
Property~S\ref{property:sorted2} of $M$, the minima of the $\theta$ or
less non-dummy rows of $M_i$ can be computed in $O(m_i+\theta)$ time
by a right-to-left and bottom-up traversal of the lower-left boundary
of the finite entries.  See Figure~\ref{figure:figure7} for an
illustration.  The row minima of $M$ can be obtained from the row
minima of the non-dummy rows of the $\beta$ matrices $M_i$ with
$m_i\geq 1$ in $O(n+\theta \cdot \beta)$ time.  The row-minima problem
on $M$ can thus be solved in $O(n+m+\theta\cdot\beta)$ time.  The
lemma is proved.
\end{proof}

For any positive integer $h$, the $j$-th column of $M$ is {\em
  $h$-brushed} if interval $[a_M(j),b_M(j)]$ contains at least one
integral multiple of~$h$. It takes $O(1)$ time to determine from
$a_M(j)$ and $b_M(j)$ whether the $j$-th column of $M$ is $h$-brushed
or not.

\begin{lemma}
\label{lemma:lemma7}
If $M$ is an $n\times m$ matrix whose columns are all $h$-brushed,
then the row-minima problem on $M$ can be reduced in $O(n+m)$ time to
the row-minima problem on an $O(\frac{n}{h})\times m$ matrix $M^*$
with
$\textit{thickness}(M^*)=O(\frac{1}{h}\cdot\textit{thickness}(M))$ and
$\textit{broadness}(M^*)=O(\frac{n}{h})$.
\end{lemma}

\begin{proof}
Let $M_1$, $M_2$, and $M_3$ be the following three $n\times m$
matrices, obtainable from $M$ in $O(m)$ time, whose entry-wise minimum
is $M$. For each $i=1,2,\ldots,n$ and each $j=1,2,\ldots,m$, let
\begin{eqnarray*}
M_1(i,j)&=&
\left\{
\begin{array}{ll}
M(i,j)&\mbox{if $a_M(j)\leq i \leq 
h\cdot\left\lceil\frac{a_M(j)}{h}\right\rceil$}\\
\infty&\mbox{otherwise}
\end{array}
\right.\\
M_2(i,j)&=&
\left\{
\begin{array}{ll}
M(i,j)&\mbox{if $h\cdot\left\lceil\frac{a_M(j)}{h}\right\rceil+1\leq i \leq h\cdot\left\lfloor\frac{b_M(j)}{h}\right\rfloor$}\\
\infty&\mbox{otherwise}
\end{array}
\right.\\
M_3(i,j)&=&
\left\{
\begin{array}{ll}
M(i,j)&\mbox{if $h\cdot\left\lfloor\frac{b_M(j)}{h}\right\rfloor+1\leq i \leq b_M(j)$}\\
\infty&\mbox{otherwise}.
\end{array}
\right.
\end{eqnarray*}
See Figure~\ref{figure:figure3} for an example.  Since each
$b_{M_1}(j)$ with $1\leq j\leq m$ is an integral multiple of $h$, we
have $\textit{broadness}(M_1)=O(\frac{n}{h})$.  Since each
$a_{M_3}(j)-1$ with $1\leq j\leq m$ is an integral multiple of $h$, we
have $\textit{broadness}(M_3)=O(\frac{n}{h})$.  By
Lemma~\ref{lemma:lemma6} with $\textit{thickness}(M_1)=O(h)$ and
$\textit{thickness}(M_3)=O(h)$, the row-minima problems on $M_1$ and
$M_3$ can be solved in $O(n+m)$ time.  Every $h$ consecutive rows of
$M_2$ are identical. Specifically, for each positive index $t$, rows
$(t-1)\cdot h+1, (t-1)\cdot h+2,\ldots,t\cdot h$ of $M_2$ are
identical. Let $M_2$ be condensed into an $O(\frac{n}{h})\times m$
matrix $M^*$ by merging every $h$ consecutive rows of $M_2$ into a
single row.  We have
$\textit{thickness}(M^*)=O(\frac{1}{h}\cdot\textit{thickness}(M))$ and
$\textit{broadness}(M^*)=O(\frac{n}{h})$.  The row minima of $M_2$ can
be obtained from those of $M^*$ in $O(n)$ time.  The lemma is proved.
\end{proof}

We are ready to prove Lemma~\ref{lemma:lemma4}.
\begin{proof}[Proof of Lemma~\ref{lemma:lemma4}]
Let $M$ be the input $n\times m$ matrix.  We first apply
Lemma~\ref{lemma:lemma5} to have $M$ sorted in $O(n+m)$ time.  Let
$\ell=1+\left\lceil\log_2\log_2 n\right\rceil$.  Assume $n\geq 2$
without loss of generality, so $\ell\geq 1$.  Define a decreasing
sequence $h_0,h_1,\ldots,h_{\ell}$ of positive integers as follows.
\begin{displaymath}
h_k=\left\{
\begin{array}{ll}
2^{2^{\ell-k-1}}&\mbox{if $0\leq k\leq \ell-1$}\\
1&\mbox{if $k=\ell$}.
\end{array}
\right.
\end{displaymath}
Each $h_k$ is a power of two. One can verify that $h_0\geq n$,
$h_1<n$, $h_{\ell-1}=2$, and $h_{k-1}=h_k^2$ holds for each
$k=1,2,\ldots,\ell-1$.  For each $k=1,2,\ldots,\ell$, if $k$ is the
smallest positive integer such that the $j$-th column of $M$ is
$h_k$-brushed, then let $j\in J_k$.  By $h_{\ell}=1$, sets
$J_1,J_2,\ldots,J_{\ell}$ form a disjoint partition of the indices of
the non-dummy columns of $M$.
For the matrix in Figure~\ref{figure:figure7} with $n=17$, we have
$\ell=4$, $h_0=256$, $h_1=16$, $h_2=4$, $h_3=2$, $h_4=1$, $J_4=\{1\}$,
$J_3=\{2,3\}$, $J_2=\{4,5,6,7\}$, and $J_1=\{8,9,10\}$.
For each $k=1,2,\ldots,\ell$, let $j_k=|J_k|$.  By
$j_1+j_2+\cdots+j_{\ell}=m$, the lemma follows immediately from the
following two statements.
\begin{enumerate}[\em {Statement}~1:]
\item Sets $J_1,J_2,\ldots,J_{\ell}$ can be obtained from $M$ in
  $O(m+n\cdot \ell)$ time.

\item For each $k=1,2,\ldots,\ell$, the row-minima problem on the
  submatrix of $M$ induced by the columns with indices in $J_k$ can be
  solved in $O(n+j_k)$ time.
\end{enumerate}

Statement~1.  For each $i=1,2,\ldots,n$, let $M_i$ be the submatrix of
$M$ induced by the columns whose indices $j$ satisfy $a_M(j)=i$.  Let
$m_i$ be the number of columns in $M_i$.  For each $i=1,2,\ldots,n$
and each $j=1,2,\ldots,m_i$, let $\kappa(i,j)$ be the index $k$ such
that $J_k$ contains the index of the column of $M$ that is the $j$-th
column of $M_i$.  Let $\kappa(i,0)=\ell$.
Since $h_1,h_2,\ldots,h_k$ are all integral multiples of $h_k$ for
each $k=1,2,\ldots,\ell$, Property~S\ref{property:sorted2} of $M$
implies $\kappa(i,0)\geq \kappa(i,1)\geq\cdots\geq \kappa(i,m_i)\geq
1$. For each $j=1,2,\ldots,m_i$, to determine $\kappa(i,j)$, it
suffices to look for the first integer $k$ starting from
$\kappa(i,j-1)$ down to $1$ such that the $j$-th column of $M_i$ is
$h_k$-brushed but not $h_{k-1}$-brushed.  Therefore, it takes overall
$O(m_i+\ell)$ time to compute indices
$\kappa(i,1),\kappa(i,2),\ldots,\kappa(i,m_i)$.  Sets
$J_1,J_2,\ldots,J_{\ell}$ can thus be obtained in $O(m+n\cdot\ell)$
time. Statement~1 is proved.

Statement~2.  Let $M_k$ be the submatrix of $M$ induced by the columns
with indices in $J_k$.  If $j\in J_k$, then the $j$-th column of $M$
is not $h_{k-1}$-brushed, implying $\textit{thickness}(M_k) < h_{k-1}
= O(h_k^2)$.  By Lemma~\ref{lemma:lemma7}, the row-minima problem on
$M_k$ can be reduced in $O(n+j_k)$ time to the row-minima problem on
an $O(\frac{n}{h_k})\times j_k$ matrix $M_k^*$ with
$\textit{thickness}(M_k^*)=O(\textit{thickness}(M_k)\cdot\frac{1}{h_k})=O(h_k)$
and $\textit{broadness}(M_k^*)=O(\frac{n}{h_k})$.  By
Lemma~\ref{lemma:lemma6}, the row minima of $M_k^*$ can be computed in
time $O(\frac{n}{h_k}+j_k+h_k\cdot \frac{n}{h_k})=O(n+j_k)$.
Therefore, the row minima of $M_k$ can be computed in $O(n+j_k)$ time.
Statement~2 is proved. The lemma is proved.
\end{proof}

\subsection{A linear-time intermediate algorithm for matrices with very few rows}
\label{subsection:step2}
This subsection proves the following lemma.  

\begin{lemma}
\label{lemma:lemma8}
Let $n$ be a given positive integer. Let $h=\max(1,\lceil\log_2\log_2
n\rceil)$.  It takes $O(n)$ time to compute an $O(n)$-space data
structure, with which the row minima of any $h\times m$ matrix can be
computed in $O(h+m)$ time.
\end{lemma}

\begin{proof} 
Let $z$ be a binary string.  For each index $i\geq 1$, let $z(i)$
denote the $i$-th bit of $z$.  Let $\textit{pred}(z,i_2)$ be the
largest index $i_1$ with $i_1\leq i_2$ and $z(i_1)=1$.  Let $Z$
consist of all $h$-bit binary strings. By $|Z|=2^h=O(\log n)$, it
takes $o(n)$ time to construct an $o(n)$-space data structure capable
of supporting each update to $z(i)$ and each query
$\textit{pred}(z,i)$ in $O(1)$ time.  

Let $M$ be the input $h\times m$ matrix.  Subscripts $M$ of $a_M$,
$b_M$, and $c_M$ are omitted in the proof.  To avoid
{\colora boundary conditions}, \reviewa{5}\label{a5}
let there be two additional dummy rows $0$ and $h+1$ in
$M$.  We first apply Lemma~\ref{lemma:lemma5} to have $M$ sorted in
$O(h+m)$ time. The proof needs only Property~S\ref{property:sorted1}
of $M$, though.  The algorithm proceeds iteratively, one iteration per
column of $M$, obtaining $\mu(i)=\min\{M(i,1),M(i,2),\ldots,M(i,j)\}$
for all row indices $i=1,2,\ldots,h$ at the end of the $j$-th
iteration.  As a result, at the end of the algorithm, we have the
minimum of each row of $M$ computed in the {\em minima array}
$\mu$. To support efficient dynamic updates and queries, we cannot
afford to explicitly store each element of $\mu$.  Instead, we use an
$h$-element {\em query array}~$q$ together with an {\em auxiliary
  binary string} $z$ for $q$ to represent $\mu$ such that
$\mu(i)=q(\textit{pred}(z,i))$ holds for each row index
$i=1,2,\ldots,h$.  Observe that if $z(i)=0$, then the value of $q(i)$
{\colora does not matter}.\reviewa{6}\label{a6}  See Figures~\ref{figure:figure3}(a)
and~\ref{figure:figure8}(a) for examples of $\mu$, $q$, and $z$.  
\begin{algorithm}[t]
\begin{center}
\medskip
\smallskip
\begin{tabbing}
\quad\ \ \={\em Initialization}:\qquad\=Let $q(0)=\infty$, $z(0)=1$, and $z(1)=z(2)=\cdots=z(h+1)=0$.\\
\>{\em For-loop}:      \>For each $j=1,2,\ldots,m$, execute the following steps.\\
\>{\em Step~1}:\>\qquad\=Let $i_0=a(j)$, $i_2=b(j)+1$, and $i_1=\textit{pred}(z,i_2-1)$.\\
\>{\em Step~2}:\>\>If $c(j)\geq q(i_1)$, then proceed to the next iteration of the
for-loop.\\
\>{\em Step~3}:\>\>If $z(i_2)=0$, then let $z(i_2)=1$ and $q(i_2)=q(i_1)$.\\
\>{\em Step~4}:\>\>While $i_0\leq i_1$ and $c(j)<q(i_1)$, execute the following substep.\\
\>{\em Substep~4a}:\>\>\qquad Let $z(i_1) = 0$, $i_2 = i_1$, and
  $i_1=\textit{pred}(z,i_2-1)$.\\
\>{\em Step~5}:\>\>If $c(j)<q(i_1)$, then let $z(i_0)=1$ and $q(i_0)=c(j)$.\\
\>{\em Step~6}:\>\>If $c(j)>q(i_1)$, then let $z(i_2)=1$ and $q(i_2)=c(j)$.
\end{tabbing}
\end{center}
\caption{Computing the row minima for an $h\times m$ sorted concise
  matrix concisely represented by arrays $a$, $b$, and $c$.}
\label{algorithm:algorithm1}
\end{algorithm}
The algorithm is as shown in Algorithm~\ref{algorithm:algorithm1}.
The initial binary string $z$ has exactly one $1$-bit.  Each iteration
of the for-loop increases the number of $1$-bits in $z$ by at most
three via Steps~3,~5, and~6. Each iteration of the while-loop of
Step~4 decreases the number of $1$-bits in $z$ by exactly
one. Therefore, the overall number of times executing Substep~4a
throughout all $m$ iterations of the for-loop is $O(m)$. Since the
initialization takes $O(h)$ time, Algorithm~\ref{algorithm:algorithm1}
runs in $O(m+h)$ time. The rest of the proof ensures the correctness
of Algorithm~\ref{algorithm:algorithm1}.\reviewa{7}

For each $j = 0,1,\ldots, m$, let $\mu_j$, $z_j$, and $q_j$ be the
$\mu$, $z$, and $q$ at the end of the $j$-th iteration, respectively.
See Figure~\ref{figure:figure8}(b) for the query array~$q_j$ at the
end of the $j$-th iteration for each $j=0,1,\ldots,7$ on the matrix
$M$ in Figure~\ref{figure:figure8}(a).
\begin{figure}[t]
\begin{center}
\scalebox{0.83}{
\begin{minipage}[b]{8.5cm}
\begin{center}
\begin{tabular}{|c||c|c|c|c|c|c|c||c|c|c|}
\hline
$M$&1&2&3&4&5&6&7&$\mu$&$q$&$z$\\
\hline
\hline
1 & 8&  &  &  &  &  & & 8&8&1\\
\hline
2 &  & 3& 7&  &  &  & & 3&3&1\\
\hline
3 &  & 3& 7& 6& 6&  & & 3& &0\\
\hline
4 &  & 3& 7& 6& 6& 3& & 3& &0\\
\hline
5 &  &  & 7& 6& 6& 3&7& 3& &0\\
\hline
6 &  &  & 7& 6& 6& 3&7& 3& &0\\
\hline
7 &  &  &  &  & 6&  &7& 6&6&1\\
\hline
8 &  &  &  &  & 6&  &7& 6& &0\\
\hline
\end{tabular}\\[24pt]
(a)
\end{center}
\end{minipage}
\begin{minipage}[b]{8.5cm}
\newcommand{\ity}{\infty}
\newcommand{\fb}{\it\color{blue}}
\newcommand{\cg}{\cellcolor[gray]{0.9}}
\begin{center}
\begin{tabular}{|c||c|c|c|c|c|c|c|c|c|c|c|c|}
\hline
 &    $q_0$&    $q_1$&    $q_2$&    $q_3$&    $q_4$&    $q_5$&    $q_6$& $q_7$\\
\hline
\hline
0&   $\ity$&\cg$\ity$&   $\ity$&   $\ity$&   $\ity$&   $\ity$&   $\ity$&$\ity$\\
\hline\hline
1&         &   {\fb8}&\cg     8&        8&        8&        8&        8&     8\\
\hline
2&         &\cg$\ity$&\cg{\fb3}&\cg     3&\cg     3&        3&\cg{\fb3}&     3\\
\hline
3&         &         &         &         &         &         &         &      \\
\hline
4&         &         &         &         &         &         &         &      \\
\hline
5&         &         &   $\ity$&\cg{\fb7}&\cg{\fb6}&\cg{\fb6}&\cg      &      \\
\hline
6&         &         &         &         &         &         &         &      \\
\hline
7&         &         &         &   $\ity$&   $\ity$&\cg      &        6&     6\\
\hline
8&         &         &         &         &         &         &         &      \\
\hline\hline
9&         &         &         &         &         &   $\ity$&   $\ity$&$\ity$\\
\hline
\end{tabular}\\[8pt]
(b)
\end{center}
\end{minipage}
}
\label{a7-2}
\end{center}
\caption{(a) A sorted $8\times 7$ concise matrix $M$, the final minima array
  $\mu$ of $M$, the final query array $q$ of $\mu$, and the final auxiliary binary
  string $z$.  The $\infty$-entries of $M$ and the entries of $q$ that
  do not matter are left out.  (b) For each $j=0,1,\ldots,7$, the
  query array $q_j$ at the end of the $j$-th iteration of the
  for-loop. The entries that do not matter are left out.
  {\colora The shaded cells of the $j$-th column with $1\leq j\leq 6$ indicate the indices $i_1$ and $i_2$ in the $j$-th iteration. The italic cell of the $j$-th column indicates the index $i^*$ of the $j$-th column.
%
  For instance, we have $i_1=0$, $i^*=1$, and $i_2=2$ in the first iteration and $i_1=i^*=2$ and $i_2=5$ in the sixth iteration.}}
\label{figure:figure8}
\end{figure}
By induction on the column index $j$, we prove 
\begin{equation}
\mbox{$q_j(\textit{pred}(z_j,i)) = \mu_j(i)$ for all indices $i$ with $1\leq i\leq h$.}
\label{eq:eq4}
\end{equation}
Equation~(\ref{eq:eq4}) with $j=0$ for all indices $i$ with $1\leq
i\leq h$ follows immediately from the initialization of
Algorithm~\ref{algorithm:algorithm1}.  Assuming
\begin{equation}
\mbox{$q_{j-1}(\textit{pred}(z_{j-1},i)) = \mu_{j-1}(i)$ for all
  indices $i$ with $1\leq i\leq h$}
\label{eq:eq5}
\end{equation}
holds with $j\geq 1$, we show Equation~(\ref{eq:eq4}) by the following
analysis on the $j$-th iteration of the for-loop.
By Property~S\ref{property:sorted1}
of $M$, we have $a(j)\geq \max\{a(1),a(2),\ldots,a(j-1)\}$, implying 
\begin{equation}
\mu_{j-1}(a(j))\leq 
\mu_{j-1}(a(j)+1)\leq 
\mu_{j-1}(a(j)+2)\leq \cdots \leq
\mu_{j-1}(b(j)).
\label{eq:eq6}
\end{equation}
We first consider the case with $\mu_{j-1}(b(j))\leq c(j)$.
  See
iteration~$7$ of the example in Figure~\ref{figure:figure8} for an
instance of this situation.  By Equation~(\ref{eq:eq6}), the $j$-th
column of $M$ does not affect the content of the minima array, i.e.,
$\mu_{j}=\mu_{j-1}$.  By Equation~(\ref{eq:eq5}), at the end of
Step~1, we have
$q(i_1)=q_{j-1}(\textit{pred}(z_{j-1},b(j)))=\mu_{j-1}(b(j))\leq
c(j)$.  Therefore, Step~2 proceeds to the next iteration without
altering the content of $q$ and $z$.  By $\mu_j=\mu_{j-1}$,
$z_j=z_{j-1}$, and $q_j=q_{j-1}$, Equation~(\ref{eq:eq4}) follows from
Equation~(\ref{eq:eq5}). The rest of the proof assumes
$c(j)<\mu_{j-1}(b(j))$, implying that Steps~4,~5, and~6 are executed
in the $j$-th iteration.

To prove Equation~(\ref{eq:eq4}) for indices $i$ with $b(j)<i\leq h$,
we first show that Steps~4,~5, and~6 do not alter the values of $z(i)$
and $q(i)$ for indices $i$ with $b(j)<i\leq h$.  At the end of Step~3,
condition $c(j)<q(i_1)$ holds.  Step~6 sets $z(i_2)=1$ and
$q(i_2)=c(j)$ only if $c(j)>q(i_1)$, implying that Substep~4a executes
at least once. We have $i_2\leq b(j)$ when Step~6 alters the values of
$z(i_2)$ and $q(i_2)$.  Observe that $\max(i_0,i_1)\leq b(j)$ holds
throughout the $j$-th iteration.  Therefore, Steps~4,~5, and~6 do not
alter the values of $q(i)$ and $z(i)$ for indices $i$ with $b(j)<i\leq
h$. By Equation~(\ref{eq:eq5}) and Step~3, we have $z_j(b(j)+1)=1$ and
$q_j(b(j)+1)=\mu_{j-1}(b(j)+1)=\mu_j(b(j)+1)$.  Since
$\mu_j(i)=\mu_{j-1}(i)$ holds for indices $i$ with $b(j)<i\leq h$,
Equation~(\ref{eq:eq4}) for indices $i$ with $b(j)<i\leq h$ follows
from Equation~(\ref{eq:eq5}) for indices $i$ with $b(j)<i\leq h$.
See iterations~$1$--$6$ of the example in Figure~\ref{figure:figure8}
for instances of this situation: Step~3 alters the content of $q$ and
$z$ in iterations~$1$--$3$ and $5$--$6$; Step~3 does not alter the
content of $q$ and $z$ in iteration~$4$.

It remains to prove Equation~(\ref{eq:eq4}) for indices $i$ with
$1\leq i\leq b(j)$.  After Step~1, we have $i_0=a(j)$ for the rest of
the $j$-th iteration.  Step~4 sets $z(i)=0$ for each index $i$ with
$i_0\leq i\leq b(j)$, $z_{j-1}(i)=1$, and $c(j)<q_{j-1}(i)$.  The
following equations hold for the fixed values of indices $i_1$ and
$i_2$ after Step~4 (i.e., during the execution of Steps~5 and~6):
\begin{eqnarray}
i_0>i_1&\mbox{or}&c(j)\geq\mu_{j-1}(i_1)\label{eq:eq7}\\
\mu_{j-1}(i)&=&q_{j-1}(i_1) \mbox{\qquad\qquad\qquad\qquad for all indices~$i$ with $i_1\leq
  i\leq i_2-1$.}
\label{eq:eq8}
\end{eqnarray}
Equation~(\ref{eq:eq7}) is by the fact that the condition of
while-loop of Step~4 does not hold.  Equation~(\ref{eq:eq8}) follows
from Equation~(\ref{eq:eq5}) and $i_1=\textit{pred}(z_{j-1},i_2-1)$,
as ensured by Step~1 and Substep~4a.
By $c(j)<\mu_{j-1}(b(j))$, we have $\mu_j(b(j))=c(j)$.  Moreover, if
$i_2\leq b(j)$ (i.e., Substep~4a being executed at least once in the
$j$-th iteration), then Equation~(\ref{eq:eq6}) implies
\begin{equation}
\mbox{$\mu_j(i)=c(j)$ for all indices $i$ with $i_2\leq i\leq b(j)$.}
\label{eq:eq9}
\end{equation}
Let $i^*$ be the smallest index with $i_1\leq i^*$ and
$\mu_j(i^*)=\mu_j(i^*+1)=\cdots=\mu_j(b(j))=c(j)$.  
In iterations~$1$--$6$ of the example in Figure~\ref{figure:figure8},\reviewa{7}\label{a7-1}
for each $j=1,2,\ldots,6$, the $i^*$-th entry of $q_j$ is italic
and the $i_1$-th and $i_2$-th entries of $q_j$ with $i_1<i_2$ are
shaded in Figure~\ref{figure:figure8}(b).  For instance, we have
$(i_1,i_2,i^*)=(0,2,1)$ 
in iteration~$1$ and $(i_1,i_2,i^*)=(2,5,2)$
in iteration~6.
One can verify 
\begin{equation}
\mbox{$\mu_j(i)=\mu_{j-1}(i)$ for all indices $i$ with $1\leq
i<i^*$}
\label{eq:eq10}
\end{equation}
as follows.  For each index $i$ with $1\leq i<i_0$, we already have
$\mu_j(i)=\mu_{j-1}(i)$, since the $(i,j)$-entry of $M$ is $\infty$.
Therefore, it remains to consider the case with $i_0\leq i^*-1$ and
verify Equation~(\ref{eq:eq10}) for indices $i$ with $i_0\leq i\leq
i^*-1$.  By Equation~(\ref{eq:eq6}), it suffices to ensure
$\mu_j(i^*-1)=\mu_{j-1}(i^*-1)$.  Assume
$\mu_j(i^*-1)\ne\mu_{j-1}(i^*-1)$ for a contradiction.  We have
$\mu_{j-1}(i^*-1)>\mu_j(i^*-1)=c(j)$.  By $\mu_j(i^*-1)=c(j)$ and the
definition of $i^*$, we have $i^*=i_1$, which implies $i_0<i_1$.  By
$i_0<i_1=i^*$ and Equation~(\ref{eq:eq7}), we have $c(j)\geq
\mu_{j-1}(i_1)=\mu_{j-1}(i^*)$, implying $\mu_{j-1}(i^*-1)>c(j)\geq
\mu_{j-1}(i^*)$. By definition of $i^*$, we have $i^*\leq b(j)$.
However, $\mu_{j-1}(i^*-1)>\mu_{j-1}(i^*)$ and $i_0\leq i^*-1<b(j)$
contradict with Equation~(\ref{eq:eq6}).

Assume $i_2<i^*$ for a contradiction. By definition of $i^*$, we have
$i_2\le b(j)$, implying that Step~4a is executed at least once. By
Equation~(\ref{eq:eq9}), $\mu_j(i)=c(j)$ holds for all indices~$i$
with $i_2\leq i\leq b(j)$, which contradicts with the definition of
$i^*$.  By $i^*\leq i_2$, we have $q(i)=q_{j-1}(i)$ and
$z(i)=z_{j-1}(i)$ for all indices~$i$ with $1\leq i<i^*$ at the end of
Step~4.  By $i_1\leq i^*$, we have $z(i)=0$ for all indices~$i$ with
$i^*<i\leq b(j)$ at the end of Step~4. Combining with
Equation~(\ref{eq:eq10}), in order to satisfy Equation~(\ref{eq:eq4})
for all indices~$i$ with $1\leq i\leq b(j)$, it suffices for Steps~5
and~6 to additionally ensure $z(i^*)=1$ and $q(i^*)=c(j)$. By the
following case analysis, ensuring $z_j(i^*)=1$ and $q_j(i^*)=c(j)$ is
exactly what Steps~5 and~6 do.
\begin{itemize}
\item {\em Case~0:} $c(j)<q_{j-1}(i_1)$.  We show $i^*=i_0$. By
  $c(j)<q_{j-1}(i_1)=\mu_{j-1}(i_1)$ and Equation~(\ref{eq:eq7}), we
  have $i_1<i_0$.  Before executing Step~4, we have $i_0<i_2$.  Each
  time when Substep~4a is executed, the current value of $i_2$ equals
  the value of $i_1$ in the previous iteration of the while-loop, when
  condition $i_0\leq i_1$ of the while-loop must hold. No matter
  whether Step~4a is executed or not, we have $i_0\leq i_2$ at the end
  of Step~4.  If $i_0<i_2$, then $i_1<i_0<i_2$ and
  Equation~(\ref{eq:eq8}) imply $\mu_{j-1}(i_0)=q_{j-1}(i_1)>c(j)$.
  If $i_0=i_2$, then~$i_0$ equals the value of $i_1$ at the execution
  of Substep~4a for the last time, when condition $c(j)<q(i_1)$ of the
  while-loop must hold. Thus, we have
  $\mu_{j-1}(i_0)=q_{j-1}(i_0)>c(j)$.  Either way, we have
  $\mu_{j-1}(i_0)>c(j)$. By $\mu_{j-1}(i_0)>c(j)$, and
  Equation~(\ref{eq:eq6}), we have $\mu_{j}(i)=c(j)$ 
  for all indices~$i$ with $i_0\leq i\leq b(j)$. By $i_1<i_0$, we have
  $i^*\leq i_0$.  By $i_1\leq i_0-1<i_2$ and Equation~(\ref{eq:eq8}),
  we have $\mu_j(i_0-1)=\mu_{j-1}(i_0-1)=q_{j-1}(i_1)>c(j)$, implying
  $i^*=i_0$.

\item {\em Case~1:} $c(j)=q_{j-1}(i_1)$.  We show $i^*=i_1$.  By
  $c(j)=q_{j-1}(i_1)$ and the fact that condition $c(j)<q(i_1)$ holds
  at the end of Step~3, we know that Step~4a is executed at least
  once, implying $i_2\leq b(j)$ and Equation~(\ref{eq:eq9}).  By
  $c(j)=q_{j-1}(i_1)$ and Equation~(\ref{eq:eq8}), we have
  $\mu_{j-1}(i)=c(j)$ and thus $\mu_j(i)=c(j)$ for all indices $i$
  with $i_1\leq i<i_2$.  Therefore, $i^*=i_1$.  

\item {\em Case~2:} $c(j)>q_{j-1}(i_1)$.  We show $i^*=i_2$.  By
  $c(j)>q_{j-1}(i_1)$ and the fact that condition $c(j)<q(i_1)$ holds
  at the end of Step~3, we know that Step~4a is executed at least
  once. By Equation~(\ref{eq:eq9}), we have $i^*\leq i_2$.  By
  Equation~(\ref{eq:eq8}) and $c(j)>q_{j-1}(i_1)$, we have
  $c(j)>\mu_{j-1}(i_2-1)$, implying $\mu_j(i_2-1)<c(j)$.  Therefore,
  $i^*=i_2$.  
\end{itemize}
For Case~0, i.e., $i^*=i_0$, as illustrated by iterations~$1$ and~$2$
of the example in Figure~\ref{figure:figure8}, Step~5 correctly sets
$z_j(i^*)=1$ and $q_j(i^*)=c(j)$.  For Case~2, i.e., $i^*=i_2$, as
illustrated by iterations~$3$ and~$4$ of the example in
Figure~\ref{figure:figure8}, Step~6 correctly sets $z_j(i^*)=1$ and
$q_j(i^*)=c(j)$.  For Case~1, we have $i^*=i_1$, as illustrated by
iterations~$5$ and~$6$ of the example in Figure~\ref{figure:figure8}.
At the end of Step~4, we already have $z(i^*)=1$ and $q(i^*)=c(j)$.
Since Steps~5 and~6 do not alter the content of $q$ and $z$, we also
have $z_j(i^*)=1$ and $q_j(i^*)=c(j)$.  The lemma is proved.
\end{proof}

\subsection{Proving Lemma~\ref{lemma:lemma3}}
\label{subsection:step3}
We are ready to prove the lemma of the section.

\begin{proof}[Proof of Lemma~\ref{lemma:lemma3}]
It suffices to prove the lemma for the case that the input $n\times m$
matrix is concise.  Let $h=\max(1,\lceil\log_2\log_2 n\rceil)$.  
Let $M$ be the submatrix of the
input matrix induced by the $h$-brushed columns.  By
Lemma~\ref{lemma:lemma7}, the row-minima problem on $M$ can be reduced
in $O(n+m)$ time to the row-minima problem on an $O(\frac{n}{h})\times
O(m)$ matrix $M^*$.  By Lemma~\ref{lemma:lemma4}, the row minima of
$M^*$ can be computed in time $O(\frac{n}{h}\log\log n+m)=O(n+m)$,
which yield the row minima of $M$ in $O(n+m)$ time.

Let $M_0$ be the submatrix of the input matrix induced by the columns
that are not $h$-brushed.  Let $\ell=\lceil \frac{n}{h}\rceil$.  For
each $k=1,2,\ldots,\ell$, let $M_k$ be the submatrix of $M_0$ induced
by the columns whose indices $j$ satisfy $(k-1) \cdot h <
a_{M_0}(j)\leq b_{M_0}(j) < k\cdot h$ and the rows with indices
$(k-1)\cdot h+1, (k-1)\cdot h+2, \ldots, k\cdot h-1$.  See
Figure~\ref{figure:figure4} for an illustration.  Let $m_k$ be the
number of columns in $M_k$.  By Lemma~\ref{lemma:lemma8}, the row
minima of $M_k$ can be computed in $O(h+m_k)$ time, with the help of
an $O(n)$-time pre-computable data structure.  As a result, the
row-minima problems on all matrices $M_k$ with $1\leq k\leq \ell$ can
be solved in overall time $O(n)+\sum_{1\leq
  k\leq\ell}O(h+m_k)=O(n+m)$.  The row minima of $M_0$ can be obtained
from combining the row minima of $M_1,M_2,\ldots,M_\ell$ in $O(n+m)$
time. The lemma is proved.
\end{proof}

\section{Concluding remarks}
\label{section:conclude}

For directed acyclic graphs and undirected graphs, we give linear-time
reductions for the replacement-paths problem to the single-source
shortest-paths problem.  The reductions are based upon our
$O(n+m)$-time algorithm for the row-minima problem on an
$O(1)$-concise $n\times m$ matrix, which is allowed to have negative
entries.  On the one hand, our reductions for directed acyclic graphs
in~\S\ref{subsection:reduction-edge}
and~\S\ref{subsection:reduction-node} work even if there are
negative-weighted edges.  Therefore, we have shown that the
replacement-paths problem on directed acyclic graphs with general
weights is no harder than the single-source shortest-paths problem on
directed acyclic graphs with general weights.  On the other hand, our
reductions for undirected graphs in~\S\ref{subsection:reduction-edge}
and~\S\ref{subsection:reduction-node} do assume nonnegativity of edge
weights.  However, it is not difficult to accommodate
negative-weighted edges in undirected graphs for the replacement-paths
problem as to be briefly explained in the next two paragraphs.

\begin{figure}[t]
\centerline{\scalebox{0.83}{\small\input{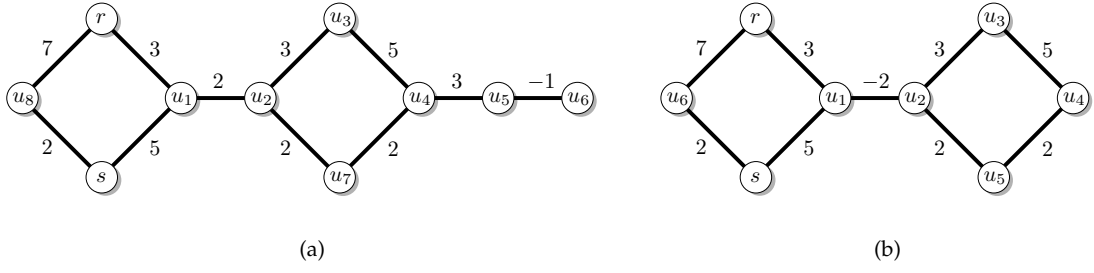}}}
\caption{Two undirected connected graphs $G$ with $d_G(r,s)=-\infty$.}
\label{figure:figure9}
\end{figure}

Let $r$ and $s$ be two nodes of the input connected undirected
$n$-node $m$-edge graph $G$ with negative-weighted edges. 
See
Figure~\ref{figure:figure9} for examples.  
We have
$d_G(r,s)=-\infty$. 
{\colora $G$ has no shortest $rs$-path.}\reviewa{8}\label{a8}
The input $rs$-path $P$ must
pass some negative-weighted edge an infinite number of times.  
For each edge $e\in P$, let
$G_e$ denote the connected component of $G-e$ that contains~$r$.  It
takes overall $O(n+m)$ time to classify all edges $e$ of $P$ into the
following three sets.
\begin{itemize}
\item {\em Set~1:}
$s\notin G_e$. We have $d_{G-e}(r,s)=\infty$.  
\item {\em Set~2:}
$s\in G_e$ and $G_e$ has negative-weighted edges. We have
$d_{G-e}(r,s)=-\infty$.
\item {\em Set~3:}
$s\in G_e$ and $G_e$ has no negative-weighted edges. We have
$d_{G-e}(r,s)=d_{G_e}(r,s)$.
\end{itemize}
It can be verified that if Set~3 is non-empty, then distances
$d_{G_e}(r,s)$ are identical for all edges $e$ of Set~3.
See~Figure~\ref{figure:figure9}(a) for an example.  The edges in
Set~3 are $u_1u_2$, $u_4u_5$, and $u_5u_6$.  We have
$d_{G-u_1u_2}(r,s)= d_{G-u_4u_5}(r,s)= d_{G-u_5u_6}(r,s)=8$.
Therefore, the replacement-paths problem on $G$ with respect to $P$
can be reduced in $O(n+m)$ time to the single-source shortest-paths
problem on $G_e$ for an arbitrary edge $e$ in Set~3.
As a result, the edge-avoiding version of the replacement-paths
problem on undirected graphs with general weights is no harder than
the single-source shortest-paths problem on undirected graphs with
{\em nonnegative} weights.  

The node-avoiding version of the replacement-paths problem is slightly
more complicated.  For each node $v\in P$ other than $r$ and $s$, let
$G_v$ denote the connected component of $G-v$ that contains $r$.  It
takes overall $O(n+m)$ time to classify all nodes $v$ of $P$ other
than $r$ and $s$ into the following three sets.
\begin{itemize}
\item {\em Set~1':}
$s\notin G_v$. We have $d_{G-v}(r,s)=\infty$.  
\item {\em Set~2':}
$s\in G_v$ and $G_v$ has negative-weighted edges. We have
$d_{G-v}(r,s)=-\infty$.
\item {\em Set~3':}
$s\in G_v$ and $G_v$ has no negative-weighted edges. We have
$d_{G-v}(r,s)=d_{G_v}(r,s)$.
\end{itemize}
If Set~3' is non-empty, then $d_{G_v}(r,s)$ are not necessarily
identical for all nodes $v$ of Set~3'.
See~Figure~\ref{figure:figure9}(b) for an example.  The nodes in
Set~3' are $u_1$ and $u_2$.  We have $d_{G-u_1}(r,s)=9$ and
$d_{G-u_2}(r,s)=8$.  However, one can show that there are at most two
distinct values of $d_{G_v}(r,s)$ for all nodes $v$ of Set~3'.
Therefore, the node-avoiding version of the replacement-paths
problem on undirected graphs with general weights is also no harder than
the single-source shortest-paths problem on undirected graphs with
{\em nonnegative} weights.  

Our presentation focuses on computing the edge-avoiding and
node-avoiding distances. It is not difficult to additionally report
their corresponding edge-avoiding and node-avoiding shortest paths in
$O(1)$ time per edge.  For instance, given a shortest-paths tree $T$
of $G$ rooted at $r$ and a shortest-paths tree $T'$ of $G'$ rooted at
$s$ as defined in~\S\ref{subsection:reduction-edge}, if the $xy$-th
column of the edge-replacement matrix $M$ contains the minimum of the
$i$-th row, then the union of (a) the $rx$-path in $T$, (b) the edge
$xy$, and (c) the $ys$-path in $T'$ is a shortest $rs$-path in
$G-e_i$.  The node-avoiding shortest $rs$-path can be similarly
obtained from $T$, $T'$, and a shortest-paths tree $T_0$ of $G_0$
rooted at $r_0$ as defined in~\S\ref{subsection:reduction-node}.

It would be of interest to see results for the single-source,
all-pairs, or near-optimal version of the problem of finding
replacement paths in undirected graphs or directed acyclic graphs that
avoid multiple failed nodes or edges.

\bibliographystyle{abbrv}
\bibliography{replace}
\end{document}